\documentclass[a4paper,11pt]{article}
\usepackage{fullpage}
\usepackage[USenglish]{babel}
\usepackage{lmodern}
\usepackage[T1]{fontenc}
\usepackage{microtype,xspace}
\usepackage{multirow}
\usepackage{amsfonts,amsmath,amsthm,amssymb,nicefrac}
\usepackage{graphicx}
\graphicspath{{./figures/}}
\usepackage[pagewise]{lineno}
\usepackage{color}
\usepackage[hidelinks]{hyperref}
\usepackage[small]{caption}

\newtheorem{theorem}{Theorem}
\newtheorem{lemma}[theorem]{Lemma}
\newtheorem{corollary}[theorem]{Corollary}

\newtheorem{observation}{Observation}

\newcommand{\Acal}[0]{\ensuremath{{\mathcal A}}}
\newcommand{\Ccal}[0]{\ensuremath{{\mathcal C}}}
\newcommand{\eR}[0]{\ensuremath{ \mathbb R}}
\newcommand{\Pee}[0]{\ensuremath{{\mathbb P}}}
\newcommand{\Ee}[0]{\ensuremath{{\mathbb E}}}
\newcommand{\isd}[0]{\hspace{.2ex} \raisebox{-.1ex}{$=$} \hspace{-1.5ex} 
\raisebox{1ex}{{$\scriptstyle d$}} \hspace{.8ex} }
\newcommand{\eps}{\varepsilon}
\DeclareMathOperator{\Bi}{Bi}
\DeclareMathOperator{\geo}{geo}
\DeclareMathOperator{\dd}{d}
\DeclareMathOperator{\Var}{Var}
\newcommand{\leqst}[0]{\ensuremath{\leq_{\text{st}}}}
\newcommand{\ZU}[0]{\ensuremath{Z_{\text{upper}}}}
\newcommand{\ZL}[0]{\ensuremath{Z_{\text{lower}}}}
\newcommand{\remark}[3]{\textcolor{blue}{\textsc{#1 #2:}}
\textcolor{red}{\textsf{#3}}}
\newcommand{\maarten}[2][says]{\remark{Maarten}{#1}{#2}}
\newcommand{\rodrigo}[2][says]{\remark{Rodrigo}{#1}{#2}}

\newcommand{\JG}[2][says]{\remark{Joachim}{#1}{#2}}
\renewcommand{\remark}[3]{}

\newcounter{tmpthm}

\newenvironment{proofof}[1]{\vspace{\parskip}\noindent{\it Proof of #1.}}{\hspace*{\fill}$\qed$\vspace{\parskip}}

\newenvironment{proofsketch}{\noindent{\it Proof Sketch.}}{\hspace*{\fill}$\qed$\vspace{\parskip}}

\newcommand{\recpac}{\textbf{RecMess}\xspace}
\newcommand{\recmess}{\recpac}

\title {Theoretical analysis of beaconless geocast protocols in 1D\thanks{A preliminary version of this work appeared in ANALCO 2018 ~\cite{GudmundssonKLMS18}.}}

\author{Joachim Gudmundsson\thanks{School of Information Technologies, University of Sydney, {\tt joachim.gudmundsson@gmail.com}}
        \and Irina Kostitsyna\thanks{Dept. of Mathematics and Computer Science, Eindhoven University of Technology, {\tt i.kostitsyna@tue.nl}, supported in part by the Netherlands Organisation for Scientific Research (NWO) under project no. 639.023.208.}
        \and Maarten L\"offler\thanks{Dept. of Information and Computing Science, Universiteit Utrecht, {\tt m.loffler@uu.nl}, supported by the Netherlands Organisation for Scientific Research (NWO) under project no. 639.021.123 and 614.001.504.}
        \and Tobias M\"uller\thanks{Bernoulli Institute for Mathematics, Computer Science and Artificial Inteligence,  Groningen University, \tt{tobias.muller@rug.nl}.}
        \and Vera Sacrist\'an\thanks{Dept. de Matem\`atiques, Universitat Polit\`ecnica de Catalunya, {\tt\{vera.sacristan,rodrigo.silveira\}@upc.edu}, partially supported by grant PID2023-150725NB-I00 funded by MICIU/AEI/10.13039/501100011033.}
        \and Rodrigo I. Silveira\footnotemark[6]}

\date{}

\begin{document}

\maketitle

\begin{abstract}
Beaconless geocast protocols are routing protocols used to send messages in mobile ad-hoc wireless networks, in which the only information available to each node is its own location. 
Messages get routed in a distributed manner: each node uses local decision rules based on the message source and destination, and its own location. In this paper we analyze six different beaconless geocast protocols, focusing on two relevant 1D scenarios.
The selection of protocols reflects the most relevant types of protocols proposed in the literature, including those evaluated in previous computer simulations.
We present a formal and structured analysis of the  maximum number of messages that a node can receive, for each protocol, in each of the two scenarios. 
This is a measure of the network load incurred by each protocol.
Our analysis, that for some of the protocols requires an involved probabilistic analysis,  confirms behaviors that had been observed only through simulations before. 

\end{abstract}

%\clearpage

\section{Introduction}

% Introduction

\maarten{We could add figures to illustrate many concepts (e.g., geocast, different protocols, etc.), but these would have to be 2D to be clear. However, this shows that 2D is more interesting than 1D.}

In mobile ad-hoc wireless networks there is no fixed infrastructure or global knowledge about the network topology.
Nodes communicate on a peer-to-peer basis,  using only local information.
Thus messages between nodes that are not within range of each other must be sent through other nodes acting as relay stations.
An important special case of ad-hoc wireless networks are wireless sensor networks, in which a (usually large) number of autonomous sensor nodes collaborate to collectively gather information about a certain area.

Nodes are typically mobile devices whose location and availability may change frequently, resulting in a highly dynamic environment in which routing must be done on-the-fly.
Typically, messages are not sent to a particular network address, but to some or all nodes within a geographic region.
This is known as \emph{geocasting}~\cite {Maihofer04}.
The main pieces of information used to send a message are the locations of the source node and the destination region (also referred to as \emph{geocast region}), which are usually included in the actual message.\footnote{Some works use the term `packet' to denote the indivisible unit of information sent between the nodes. In this paper we use the term `message' instead, as we are interested in counting the number of transmissions and not the higher level aspects of protocols.}
See Figure~\ref{fig:intro} for an illustration.

Many geocast protocols have been proposed.
In general, existing protocols can be divided into two groups: those that assume that each node also knows the location of its 1-hop neighbors (i.e., all nodes within range) and those that do not make this assumption.
In practice, the locations of neighbors can be obtained by regularly exchanging \emph{beacon} messages in the neighborhood.
Beacons imply a significant message overhead, which prevents these methods from scaling even to medium-size networks~\cite{blr}: the problem is that in dense environments the number of messages \emph{received} by each individual node, and thus the workload to decide whether and how to react to those messages, becomes prohibitive.
For this reason, in this paper we are interested in the second group, the so-called \emph{beaconless} geocast protocols.

\begin{figure} [bht]
\includegraphics[width=8cm]{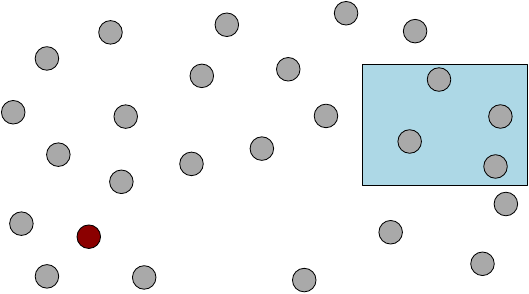}
\centering
\caption{A geocast example in 2D where a message should be sent from the sender (red node) to the geocast region marked as a blue rectangle.}
\label{fig:intro}
\end{figure}

Probably the most straightforward beaconless geocast protocol is \emph{simple flooding}: each message is broadcasted to all neighbors, who in turn broadcast it to all their neighbors, and so on.
Even though it is effective, the resulting message overhead is clearly unaffordable: essentially, it causes as much overhead as the exchange of beacons, and therefore has the same scaling problem.
From there on, there have been many improvements proposed.
The goal is to reduce the message overhead while still guaranteeing delivery. In the last few decades, many different geocast protocols have been proposed. The most important protocols are the subject of this paper, and will be described in detail in the next section.
For a thorough review of all existing beaconless geocast protocols we refer the reader to the surveys by Maih\"ofer~\cite{Maihofer04} and R\"uhrup~\cite{r-tpgr-09}.

Given the importance of geocast protocols and the many options available, there have been a few comparative studies that assessed the efficiency and efficacy of different methods under different scenarios, using low-level computer simulations.
Maih\"ofer~\cite{Maihofer04} presents simulations for four geocast protocols (flooding, two variants of flooding restricted to a forwarding zone, and a greedy-based protocol) for sparse networks with 100--1000 nodes that move 
% randomly - \JG{There are a lot of restrictions on the movement, so I would avoid the term "randomly".}
within a square region, and use random circular geocast regions as destinations. The results are analyzed in terms of total number of messages transmitted (network load) and success rate.
As expected, the experiments show that flooding has the highest success rate, but does not scale well, while the other methods are much better in terms of network load, but suffer from lower success rate.
Hall~\cite{h-igmanet-11} evaluated four methods (M heuristic, T heuristic, CD and CD-P, and several combinations of them) in 14 different scenarios, many of them based on realistic training applications.
The parameters studied were success rate and average latency.
The simulation was done using a high fidelity simulator capable of modeling realistic MAC and queuing behavior, allowing to reflect the effects of high network load in the different protocols. 
The main conclusion in~\cite{h-igmanet-11} is that CD-P performs best in terms of both success rate and latency.
Interestingly, the experiments also show that for restricted flooding heuristics, an increase in the redundancy parameters does not always lead to higher success rate due to more collisions and medium contention.

In contrast with previous comparisons, in this paper we are interested in analyzing the behavior of beaconless geocast protocols from a theoretical perspective. To that end, we present a structured overview of six  different protocols that represent the main existing protocols in the literature, and identify important quality criteria to analyze them mathematically.

The protocols analyzed are simple flooding, M heuristic~\cite{ha-tgp-06}, T heuristic~\cite{ha-tgp-06}, CD~\cite{h-igmanet-11}, CD-P~\cite{h-igmanet-11}, and delay-based protocols (that include, as particular cases, protocols like BLR~\cite{blr}, GeRaF~\cite{geraf} and GeDiR~\cite{StojmenovicL01}).
This selection of protocols reflects the main types of beaconless geocast protocols, and includes those evaluated in previous computer simulations~\cite{h-igmanet-11,Maihofer04}.\footnote {Note that, in practice, protocols are often combined: a message is forwarded if at least one protocol requires it to do so. This increases success rates, at the cost of a higher total number of messages in the system.}

Several criteria can be taken into account when comparing the behavior of different protocols.
The \emph{success rate} measures the fraction of sent messages that actually reach the target.
For those that arrive, the \emph{hop count} indicates how many steps (forwards) are needed.
In this paper we only focus on what we consider to be the most significant measure within this context: the maximum number of messages that a node receives (\recpac). 
%\rodrigo{Is packet = message? Because here we talk about packets, but the parameter is RecMess, not RecPack. We may want to be consistent, Hall uses mostly packet, but message also appears, albeit maybe with a different meaning} \irina{I may be mistaken, but I think the way it works is this: higher level protocols prepare packets that are sent to the lower level protocols that split them into smaller packets, and eventually on the bottom level the final packets are split into messages that are actually getting sent around } \irina[continues]{I suggest we stick to `messages'}
%\rodrigo{I just checked a bit more, you are aright in that it depends on the level you are, but it seems to be the other way around. At higher levels there are messages, which get split into packets. In any case, I think it's most important to stick to the terms used in previous literature. Following Hall, I would only talk about packets. For instance, he writes: ``Geocast is a network protocol for sending a packet to all nodes within a defined geographic region termed as the geocast region.''} 
%\irina{Yes! You are right. We agreed with Maarten on still using `message' everywhere, as we don't really care about high and low level protocols and splitting the messages into pieces. We'll add a note to the introduction that we assume message to be indivisible and that some works call packets what we call messages}
This parameter measures the work or energy consumption for a node, as well as the overall network load and therefore, its congestion. 
We note that network load is directly related to success rate, thus indirectly this aspect is also being considered, as done in previous comparisons~\cite{h-igmanet-11,Maihofer04}; we do consider the theoretical success probability of the protocols (in case of no collisions) separately.
We also note that, in most situations, \recpac is larger than the number of sent messages, because for intermediate nodes, the sending of a message occurs only as a consequence of receiving one before.

The behavior of a geocast protocol, in general, must be analyzed in the context of a particular geometric setting (i.e., a certain configuration of nodes and radio obstacles).
%%%
%The results in this paper are a first step towards analyzing geocast protocols in generic geometric settings.
%This paper has focused on two relatively simple scenarios.
%Between these basic scenarios and the final intricacies of real-world situations, several abstractions of the geocast problem of varying complexity can be imagined, which
%are definitely worth studying next (see Figure~\ref {fig:scenarios2}).
%%%

In this paper, we focus on two fundamental geometric scenarios in 1D: unbounded range and bounded range.
Even though it is clear that the full complexity of these protocols can only be appreciated in two dimensions, we show that the 1D scenarios considered, despite their apparent simplicity, already pose interesting challenges, and expose many of the essential differences between the protocols studied.
Moreover, understanding 1D situations is useful for many 2D scenarios as well, in which there are local situations that behave essentially as one-dimensional (see Figure~\ref{fig:1d2d}).

\paragraph{Results and paper structure.}
For each of the two scenarios, bounded and unbounded range, we analyze the worst- and expected-case performance for each protocol.
The results obtained corroborate many of the findings previously obtained only by simulations, and provide new insights into the difficulties of the 2-dimensional case where, in addition, escaping from local optima is necessary and requires combining different techniques.  
In fact, we note that the 1D setting has been used before to understand the behavior of the CD and CD-P heuristics~\cite{h-igmanet-11}.
A summary of our results is presented in Tables~\ref{tab:results_worst_case} and~\ref{tab:results_probabilistic}.

In Section~\ref{sec:protocols} we review the important geocast protocols that will be the subject of our analysis.
In Section~\ref{sec:problem_model} we discuss the scenarios considered in this paper together with model choices.
The theoretical analysis is presented in three parts.
In Section~\ref{sec:worst_analysis} we analyze the maximum number of messages that a node can receive.
In order to obtain more fine-grained bounds, in Section~\ref{sec:probabilistic_analysis} we analyze the expected number of received messages assuming fair medium access~\cite{FairMediaAccess} for two of the protocols.
The proofs of some of these results are very extensive, so they are presented separately in Sections~\ref{appx:omitted-cdp} and~\ref{appx:omitted-cd}.
The last part of the analysis, presented in Section~\ref{sec:delayed_analysis}, is devoted to analyzing the effect of delay functions in the number of received messages.
Finally, we present some concluding remarks in Section~\ref{sec:conclusions}.

\begin{figure}
\includegraphics{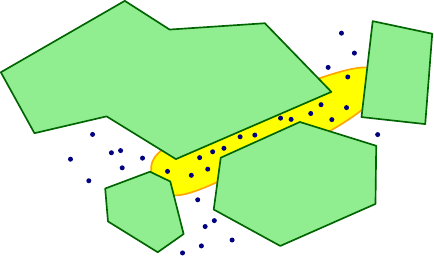}
\centering
\caption{The 1D scenario serves as a first step towards the general question, and is already relevant in situations where the local geometry is essentially 1-dimensional,  e.g., in narrow passages.
In the figure, points represent nodes and green polygons obstacles.}
\label{fig:1d2d}
\end{figure}

\begin{table}
\small\centering
\begin{tabular}{|l|c|c|}
\hline
 Protocol  & 1D unbounded reach & 1D bounded reach ($r \ll n$) \\ \hline \hline
Simple flooding & $\recmess=(n-2)k$ & $\recmess= 2rk$ \\ \hline

M-heuristic & $\recmess = Mk$ & $ \min (\frac{M}{2}, 2r)k  \leq \recpac \leq \min (M, 2r) k$ \\ \hline

T-heuristic & 
$\lceil\frac{n}{2T}\rceil k\leq\recpac\leq \lceil\frac{n}{T}\rceil k$ 
& $\lceil\frac{r}{T}\rceil k\leq\recpac\leq \lceil\frac{2r}{T}\rceil k$ \\ \hline
CD \& CD-P & 
$k\leq\recpac\leq nk$  &
$2k\leq\recpac\leq 2rk$ \\ \hline
\end{tabular}
\caption{Summary of results in worst-case analysis of activation order, where $n$ is the number of nodes, $k$ is the total number of messages, and $r$ is the communication range (refer to Section~\ref{sec:problem_model} for precise definitions).
% \JG{Should we note why ``delay-based'' are omitted from this table?}
Note that delay-based heuristics are omitted from this table as they influence the activation order, which does not affect the worst-case performance.} 
\label{tab:results_worst_case}
\end{table}

\begin{table}
\small\centering
\begin{tabular}{|l|c|c|}
\hline
 Protocol  & 1D unbounded reach & 1D bounded reach ($r \ll n$) \\ \hline \hline
CD & $ \recmess \in
\begin{cases}
 \Theta (k^2 \log (\lceil n/k \rceil + 1))\,, & \text{ if } k \le n\\
 \Theta (nk)\,, & \text{ if } k > n
\end{cases}
$ & $\recmess \in O(k^{3/2})$ \\ \hline
CD-P & $\recmess \in \Theta(k \log n)$ & $\recmess \in \Theta(k)$ \\ \hline
\end{tabular}
\caption{Summary of results in probabilistic analysis of activation order.}
\label{tab:results_probabilistic}
\end{table}

\section{Geocast protocols}
\label{sec:protocols}
% Studied protocols

We begin by describing the basic geocast framework operating at each node, following~\cite{h-igmanet-11}. 
We consider a setting with $n$ nodes represented by points in 1D,
%a geometric space, 
and $k$ messages that must be sent between the nodes. During the execution of a particular protocol, messages may be sent by the original senders, received by other nodes, and either retransmitted or deleted. 

%%%%

Each message is assumed to contain a unique ID, the location of the sender, and the destination (geocast region).
Each node has limited memory, which we assume is larger than $k$ but considerably smaller than $n$.
We assume that received messages are stored in a queue.  
When a node receives a message, it checks if it has already received a message with the same ID. 
If not, it creates a new record for the ID, and enqueues the message for potential later retransmission (possibly after storing some extra information about the message). 
When a message reaches the top of the transmission queue and is ready to be transmitted, a \emph{heuristic check} is performed. If passed, the message is
transmitted, otherwise discarded.
The main difference between the different protocols lies on how this heuristic check is performed.
In general terms, this check is a combination of local decision rules. 
Often, one of these rules is a location predicate to control the region where each message must travel (e.g., to guarantee that messages are only forwarded within a certain zone that contains the geocast region).

In the remaining of the paper we focus on the following beaconless geocast protocols, which can be categorized as: simple flooding, restricted flooding, distance-based, and delay-based.

\subsection{Simple flooding} \JG{Should we move this paragraph into Section~2.1 and change the section heading to ``Flooding \& Restricted flooding"?}\rodrigo{I turned the paragraph into a subsection, to put it at the same level as the others. I'd prefer to keep simple flooding separate from the M and T heuristics, but I don't mind if others prefer it  merged}
The simple flooding protocol works as follows.
When a node receives a message, first it checks if it has been broadcasted before. If not, then the message is broadcasted, and its ID stored in order to make sure it will not be broadcasted again.

This strategy is simple and robust, but it is non-scalable, as it produces an excessive and unnecessary network load.
In the following, we describe several heuristics intended to reduce such flood load.
Nevertheless, it is interesting to consider simple flooding not only for comparison purposes, but also because it is used as a building block in other protocols~(e.g.,~\cite{KoV99}).

\subsection{Restricted flooding}
In order to reduce the number of unnecessary transmissions of the same message, one can limit retransmissions in several ways~\cite{h-igmanet-11}. The following two heuristics apply different approaches for this: a direct limit on the number of retransmissions, or an implicit limit, by only making ``far away'' nodes retransmit already heard messages.

\paragraph{M heuristic~\cite{ha-tgp-06}.}
The MinTrans (M) heuristic explicitly controls redundancy through a parameter $M>1$:
A node broadcasts a received message if and only if the number of transmissions received for that ID is less than $M$.
The redundant propagation allowed by the parameter $M$ helps against problems such as message collisions and getting out from local optima.

\paragraph{T heuristic~\cite{ha-tgp-06}.}
The Threshold (T) heuristic uses location information for spreading the geocast propagation outward:
A node retransmits a received message if and only if the closest among all transmitters of messages with the same ID is at least a distance $T$ from it.

\subsection{Distance-based heuristics}
The restricted flooding heuristics are likely to have delivery failures in the presence of obstacles. The following protocols were designed to help solve this problem.

\paragraph{CD heuristic~\cite{h-igmanet-11}.}
The Center-Distance (CD) heuristic relies on proximity: A node retransmits a received message if and only if its distance to the center of the geocast region is less than that of all originators of transmissions received for the message ID.
This heuristic reduces some of the scalability problems of a previous method, Ko and Vaidya's ``Scheme 2''~\cite{KoV99}, which considered only the distance of the first transmitter heard of the same message instead of that of all of them.

\paragraph{CD-P heuristic~\cite{h-igmanet-11}.}
This protocol uses priority queues in order to further reduce the scalability problems of the CD heuristic.
It works as follows:
Each time the node can transmit, it transmits any message that has not been transmitted at all yet (if any) or it retransmits, among all heard messages, the one whose transmission would give the largest reduction in distance to the center of the geocast region.

\subsection{Delay-based heuristics}
\label{sec:delay_methods}
Some strategies to further reduce redundancy combine the local decision rules in the previous protocols with retransmission delays, which are based on \emph{delay functions}.
In principle, the local decision rules are independent of the delay function used, leading to a large number of possibilities.

Existing delay-based heuristics, discussed in some detail below, can be quite different, making it hard to analyze them in a unified way in our theoretical framework. In some sense, delay functions can be viewed as a way to choose a specific \emph{activation order} (see Section~\ref{sec:activation_order}) for the nodes.
%, as  discussed in Section~\ref{sec:delayed_analysis}.

%In order to be able to include them in our study, we introduce two abstract delay-based protocols that capture two important aspects of these protocols.

\paragraph{BLR heuristic~\cite{blr}.}
In the Beacon-Less Routing (BLR) protocol, each node determines when to retransmit a received message based on a dynamic forwarding delay function which returns a value in the range $[0, \text{MD}]$, where MD is a constant representing the maximum delay. The node retransmits the package after such delay, unless some other node does it before, in which case the retransmission is canceled.
Three delay functions have been suggested in~\cite{blr}, based on the following parameters: $r$ (transmission range), $p$ (progress towards destination of the orthogonal projection of the current node onto the line connecting the previous node to the destination), and $d$ (distance from current node to the source-destination line): \rodrigo{A figure illustrating all these parameters would be nice}
\[
        \text{delay}_1  = \text{MD} \cdot \frac{r-p}{r}\,, \qquad
        \text{delay}_2  = \text{MD} \cdot \frac{p}{r}\,, \qquad
        \text{delay}_3  = \text{MD} \cdot \frac{e^{\sqrt{p^2+d^2}}}{e}\,.
\]
%Methods with $\text{delay}_1$ and $\text{delay}_2$ are equivalent to preceding protocols MFR~\cite{MFR-84} and NFP~\cite{NFP-86}, respectively.

\paragraph{GeRaF heuristic~\cite{geraf}.}
%Based on distance, 
The Geometric Random Forwarding (GeRaF) protocol partitions logically the area around the destination of a message $m$ into $n_m$ areas
${\cal A}_1, \dots,{ \cal A}_{n_m}$, 
%where in ${\cal A}_1$ are all nodes closest to the destination, and so on.
such that all points in ${\cal A}_i$ are closer to the destination than any point in ${\cal A}_{i+1}$ for $1\leq  i < n_m$.

Once $m$ is transmitted, up to $n_m$ phases start, during which all nodes listen during a fixed amount of time.
In the first phase, nodes in region ${\cal A}_1$ get to reply. If only one node replies, then that one will forward the message.
If there are more, some collision resolution scheme must be used.
If there is no reply, then it is the turn to reply for nodes in region ${\cal A}_2$.
This process continues until some node in the non-empty region closest to destination replies.

\paragraph{Greedy routing (beaconless version).}
Greedy routing forwards the message to the  neighbor of the current node that is closest to the geocast region.
Even though it does not guarantee delivery, greedy routing strategies are often used as building block of geocast protocols. For this reason we also consider greedy routing in our analysis.
One example is Geographic Distance Routing (GeDiR)~\cite{StojmenovicL01}. GeDir assumes that every node knows the locations of its neighbors.
%requires to know the position of all neighbors of a node.
%: it is a greedy algorithm that always forwards the message to the neighbor of the current node whose distance to the destination is minimum.
However, it can be made beaconless by using a delay function based on the following parameters: $r$ (transmission range), $d$ (distance from previous node to destinations), and $x$ (distance from current node to destination).
$$
\text{delay}_4 = \text{MD} \cdot \frac{x+r-d}{2r}\,.
$$
This strategy tries to get out of local minima by sending the message to the best positioned neighbor, even if it is not closer to the destination.

\section{Model}
\label{sec:problem_model}

\begin{figure}[t]
  \centering
    \includegraphics{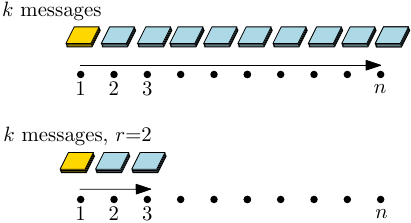}
    \caption{\small Illustration of the scenarios.
Top: with unbounded reach, the $k$ messages arrive immediately to all nodes, but that does not prevent intermediate nodes from forwarding the messages. Bottom: with range $r=2$, the messages sent from node $1$ only reach up to node $3$, so forwards are necessary to reach the target, $n$. 
%\JG{Updated the numbers from $[0,n+1]$ to $[1,n]$. Please check that it doesn't clash with anything.}
}

%    \includegraphics{scenarios}
%  \caption{\small Illustration of the scenarios. Top: with unbounded reach, the $k$ messages arrive immediately to all nodes, but that does not prevent intermediate nodes from forwarding the messages. Bottom: with range $r=2$, the messages sent from node $0$ only reach up to node $2$, so forwards are necessary to reach the target, $n+1$.}
\label{fig:scenarios}
\end{figure}

\subsection{The 1D scenarios}

In this paper we study two fundamental scenarios for a set of nodes in 1D, where the leftmost of $n$ nodes 
%\rodrigo{Note that $n$ is now the total number of nodes (used to be $n+2$)} \JG{I could only see the $n+2$ in Figure~\ref{fig:scenarios}, which I changed.} 
sends $k$ messages to the rightmost node (i.e., the geocast region only contains the rightmost node).
For simplicity, nodes are evenly spread at unit distance along the line. 

The two scenarios considered are nodes with \emph{unbounded} and \emph{bounded} communication range (see Figure~\ref{fig:scenarios}). 
In the first case, every transmitted message is assumed to arrive to all other nodes, while in the bounded reach scenario, a message transmitted by a node arrives only to the nodes at distance at most $r$.
Furthermore, it is reasonable to assume that $r \ll n$.
The unbounded case is interesting because the $n$ intermediate nodes form a dense bottleneck, a situation that can easily arise in practice (even for nodes in 2D), as illustrated in Figure~\ref{fig:1d2d}.

\subsection{Activation order}
\label{sec:activation_order}
We analyze the maximum number of messages that any node  receives (\recpac), among all nodes.
We are interested in the  possible values of \recpac.

%We first analyze the minimum and maximum values over all possible activation orders.

The value of \recpac depends not only on the protocol, but also on the order in which the nodes \emph{activate}.
Even though in theory nodes process messages instantly after their reception, in practice transmissions may cause collisions.
We abstract from the way in which collisions are handled by assuming that nodes activate in some sequential order.

First, in Section~\ref{sec:worst_analysis}, we analyze the protocols in a pessimistic setting, by taking the maximum \recpac over all possible activation orders.

Second, to achieve sharper and more realistic bounds, in Section~\ref{sec:probabilistic_analysis} we consider the special case where activation orders are drawn from a fair distribution, which is known in the literature as \emph{fair medium access}~\cite{FairMediaAccess}.
Under this model, the transmission is done in rounds, in which only one node is activated, and where each node has the same probability of being the one being activated.

Finally, in Section~\ref{sec:delayed_analysis} we analyze \recmess under delay-based heuristics, where the protocol  purposely tries to control the activation order by means of the delay function.

%For each of the two settings and each protocol studied, we analyze the maximum number of messages that a node can receive, denoted \recpac.

\subsection{Message grid representation}
\label{sec:message_grid}
The 1D setting and the assumption that nodes are activated sequentially, allow us to represent the state of the message queues of every node using an $n\times k$ table, that we call \emph{message grid}.\rodrigo{Refer to ``message grid'' in following sections as well?} 
The columns of the message grid represent nodes and the rows represent messages.
In particular, each column represents the message queue of the corresponding node.
In each round, one of the columns is chosen, and the corresponding node is activated.
The protocol used may cause or not the retransmission of one of the messages in the queue, removing it from the queue.
In addition, when the activated node transmits a message, that may cause that message to be added or deleted from the queues of other nodes, depending on the protocol.

Figures~\ref{fig:full_scenarios_ub} and ~\ref{fig:full_scenarios_b} illustrate the different protocols run on an example data set with $k=4$ and $n=6$, for each scenario.\footnote{In order to make the figures more insightful, we depict the situation as if messages may get instantly added to or deleted from the queues of non-active nodes, while according to our activation model, each column only changes when the node gets activated.}
Each row shows the first 7 rounds of a different protocol. At each round, the message grid is shown, and the column of the  activated node is highlighted in blue.

In the unbounded reach scenario (Figure~\ref{fig:full_scenarios_ub}), we assume that the initial state is the state after some start node transmitted all $k=4$ messages, which means that all columns contain the same four messages in the same order.
In the bounded reach scenario (Figure~\ref{fig:full_scenarios_b}) this means that only the first $r=2$ columns contain those messages.

All protocols are run in the same activation order.

\begin{figure}[!ht]
  \centering
    \includegraphics[width=\textwidth]{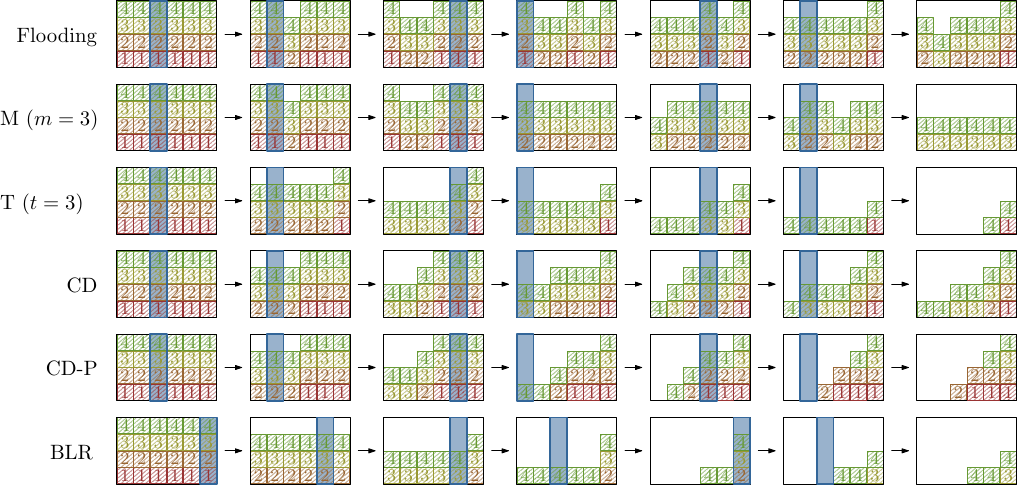}
    \caption{The different protocols run on an example data set with $k=4$ and $n=6$, for the  unbounded reach scenario. %\JG{Should we state the activation order? "
    The activation order used in the illustration for the top five protocols is $\langle 3,2,5,1,4,2 \rangle$.%"}
%Each row shows the first 7 steps of a different protocol. At each step, a matrix is shown. Columns represent the queues of the nodes, which all start containing the same four messages in the same order. A shaded column indicates the node that is about to forward a message in that step. 
%The same random order is chosen everywhere except for BLR.
}
\label{fig:full_scenarios_ub}
\end{figure}

\begin{figure}[!ht]
  \centering
    \includegraphics[width=\textwidth]{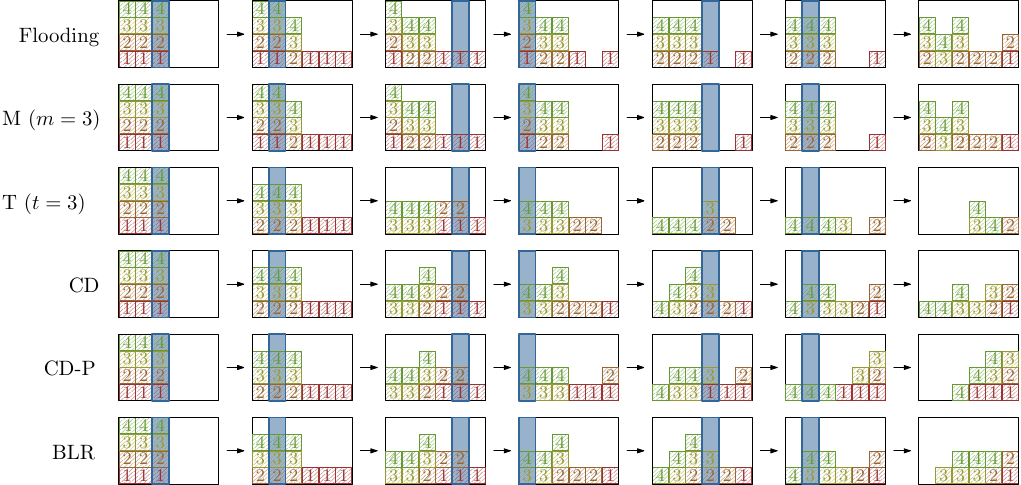}
    \caption{The different protocols run on an example data set with $k=4$ and $n=6$, for the bounded reach scenario.
    %\JG{Should we state the activation order? "
    The activation order used in the illustration %for the top five protocols 
    is $\langle 3,2,5,1,4,2 \rangle$.%"}%
    }
\label{fig:full_scenarios_b}
\end{figure}

\section{Worst-case analysis of activation order}
\label{sec:worst_analysis}

In this section we analyze the maximum (and minimum) value of \recpac over all possible activation orders, for each of the protocols, and for both scenarios.

\subsection{Simple flooding}
In both scenarios all nodes, except from the leftmost one, eventually will receive each message and therefore all nodes, except for the rightmost one, will retransmit each message.
Note that this is true independently of the activation order, and therefore $\recpac$ is independent of the activation order for this protocol.

In the unbounded reach case, all the nodes except for the rightmost will send each of the $k$ original messages, and therefore each node will receive $n-2$ copies of each message (from all but the rightmost and itself), thus $\recpac=(n-2) k$.

In the bounded reach case, every node except from the $r$ rightmost and $r$ leftmost nodes will receive each message from exactly $2r$ neighbors.
Therefore the maximum number of messages that any node can receive is $\recpac= 2r k$.

\begin{theorem}
  For the simple flooding protocol we have:
  \begin{itemize}
      \item $\recpac=(n-2) k$ in the unbounded reach scenario.
      \item $\recpac= 2r k$ in the bounded reach scenario.
  \end{itemize}
  
\end{theorem}

\subsection{M heuristic}
In the unbounded reach scenario this protocol guarantees that there are never more than $M$ copies of each messages transmitted. Therefore, every node, except for those that send the message, will receive every message $M$ times, leading to $\recpac=M k$.
This is independent of the activation order.

In the bounded reach scenario we distinguish two cases: $2r \leq M$ and $2r > M$.
In the first case the $M$ heuristic is equivalent to simple flooding, therefore it is also independent of the activation order, and $\recpac= 2r k$.

If $2r>M$, we begin by observing that, regardless of the activation order, all messages eventually reach the destination.

\begin{observation}
For the M heuristic, in the bounded reach scenario, all messages reach the destination regardless of the activation order.
\end{observation}
\begin{proof}
Consider a message $\mu$. 
At any point in time, let $v$ be the  rightmost node that has received a copy of $\mu$.
We argue that $v$ has received $\mu$ fewer than $M$ times, and therefore the message will be transmitted beyond $v$.

Suppose that $v$ has received $\mu$ $M>1$ times.
Then at least one retransmission of $\mu$ originated from a node at distance smaller than $r$ to the left of $v$.
But then this message must have reached at least one node to the right of $v$, contradicting that $v$ was the rightmost node that received a copy of $\mu$.
\end{proof}

In the worst case every node receives every message $M$ times, so $\recpac \leq M k$. We show next that every node receives each single message at least $M/2$ times.

\begin{lemma}
For the M heuristic, in the bounded reach scenario with $2r > M$, every node receives each message at least $M/2$ times, for any activation order.
\end{lemma}
\begin{proof}
Let the nodes be numbered from left to right $v_1$ to $v_n$, and consider the situation after running the whole protocol.
Consider a message $\mu$. 
First we observe that for two neighboring  nodes $v_{i}$ and $v_{i+1}$, the number of times that $v_{i}$ can have received $\mu$ is at most two times more than $v_{i+1}$.
This is because there are only two nodes that can be heard by $v_i$ but not by $v_{i+1}$, namely $v_{i-r}$ and $v_{i+1}$ itself.

Assume there is a node $v_i$ that received $\mu$ exactly $c<M$ times.
Then nodes $v_{i \pm j}$ have received $\mu$ at most $c+2j$ times.
There are $M-c$ nodes for which $c+2j < M$, all at distance smaller than $M$ from $v_i$.
Therefore, all these nodes will have retransmitted $\mu$ and $v_i$ will have received all these retransmissions, thus $c \geq M-c$. 
From this, it follows that $c \geq M/2$.
\end{proof}

From this lemma we can conclude that $\recpac \geq Mk/2$. We summarize with the following theorem.

\begin{theorem}
  For the M heuristic we have:
  \begin{itemize}
      \item $\recpac=M k$ in the unbounded reach scenario.
      \item $ \min (\frac{M}{2}, 2r)k  \leq \recpac \leq \min (M, 2r) k$ in the bounded reach scenario.
  \end{itemize}
  
\end{theorem}

\subsection{T heuristic}
In the unbounded reach scenario, due to the unit-distance distribution of the nodes and assuming $T\in\mathbb N$, when a node transmits a message, then this message is deleted from the queues of its $T$ left and its $T$ right neighbors. Thus, depending on the activation order: 
$\lceil\frac{n}{2T}\rceil k\leq\recpac\leq \lceil\frac{n}{T}\rceil k$.

In the bounded reach scenario, if $T\geq r$, no message will ever be forwarded. 
If $T<r$, then each node $u$ can receive a message from at most $2r$ nodes. Each time it receives one, at least $T$ and at most $2T$ of the nodes within reach of $u$ delete the message from their queues. Thus:
$\lceil\frac{r}{T}\rceil k\leq\recpac\leq \lceil\frac{2r}{T}\rceil k$.

%$\recpac=O(\frac{r k}{T})$.

\begin{theorem}
For the T heuristic we have:
  \begin{itemize}
      \item $\lceil\frac{n}{2T}\rceil k\leq\recpac\leq \lceil\frac{n}{T}\rceil k$ in the unbounded reach scenario.
      \item $\lceil\frac{r}{T}\rceil k\leq\recpac\leq \lceil\frac{2r}{T}\rceil k$ in the bounded reach scenario.
  \end{itemize}
\end{theorem}

\subsection{CD and CD-P heuristics} 

The number of messages received in the CD or CD-P heuristic depends heavily on the order in which the nodes are chosen. In the best case, the last node is chosen consecutively $k$ times, and every node receives only $k$ messages, since all the nodes delete the message as soon as they hear it from a node farther ahead. In the worst case, the nodes are chosen from left to right, and all nodes receive all messages $n$ times, as bad as in simple flooding.
So, 
$k\leq\recpac\leq nk$.

In the bounded reach scenario, in the best case, each message is transmitted only by the node closest to the target. Then each node receives each message only twice: once from the left and once from the right.
In the worst case, messages are retransmitted by the node farthest from the target, causing all nodes to retransmit every message. However, because of the bounded reach, the flooding effect is somewhat mitigated: each node has only $2r$ neighbors and receives each message $2r$ times.
So,
$2k\leq\recpac\leq 2rk$.

\begin{theorem}
For the CD and CD-P heuristics, we have:
  \begin{itemize}
      \item  $k\leq\recpac\leq nk$ in the unbounded reach scenario.
      \item $2k\leq\recpac\leq 2rk$ in the bounded reach scenario.
  \end{itemize}
\end{theorem}

\section{Probabilistic analysis of activation order}
\label{sec:probabilistic_analysis}

In this section we assume \emph{fair medium access}, which dictates a uniform distribution for the chosen node at each round on the set of all nodes with non-empty queues.
Under this model, we can use probabilistic analysis to make a more precise prediction of the number of received messages in practice.

Since our bounds on $\recpac$ are already tight for simple flooding and the M and T heuristics, in this section we will focus on CD and CD-P.
Although conceptually CD-P is more complicated than CD, it turns out that the analysis of CD-P is actually simpler, therefore we start from there.

\subsection{CD-P heuristic} 

\paragraph{Unbounded reach scenario}

%For this protocol, however, the number of messages reduces more quickly than in the case of CD. We show that the difference is significant and speeds up the process by roughly a factor $k$.

\begin{theorem}\label{thm:1d-clique-cdp}
  For the CD-P heuristic  in the unbounded reach scenario,   assuming fair medium access,
  %(as $k$, $n$, or both tend to infinity),
  $\recpac \in \Theta(k\log n)$ with high probability. 
\end{theorem}

\begin{proofsketch}
We present the main idea of the proof here. The full proof, including an argument for the lower bound, can be found in Section~\ref {appx:omitted-cdp}.

Initially, all $k$ messages are in the queue of each of the $n$ nodes. At each iteration of the protocol, if a node $i$ transmits a message $m$, 
it gets deleted from the queues of all the nodes to the left of $i$ and the node $i$ itself, i.e., after the iteration $m$ is present in 
the $n-i$ nodes to the right of $i$ (refer to Figure~\ref{fig:full_scenarios_b}).

When a node is selected to transmit a message, it chooses one based on the distance to the closest node it heard each message from. 
This means that at each iteration, the message currently present in the maximum number of nodes will be retransmitted.

Let $\mathcal{R}=(R_1,R_2,\dots,R_k)$, where $R_i$ denotes the number of nodes that at the current round contain message $i$ in their queues.
For a specific round or time $t$, let $R^t_i$ be the number of nodes that contain message $i$ at round $t$.
Initially, $R_1=R_2=\dots=R_k=n$. 

Let, at some iteration of the protocol, message $i$ be present in the queues of the maximum number of 
nodes, i.e., $R_i=\max \mathcal{R}$. Due to the assumption on the fair medium access, a retransmitting node is chosen uniformly at random. 
Therefore, after the iteration, $R_i$ is replaced with a random value uniformly distributed on $\{1,\dots,R_i-1\}$. 
We are interested in the number of iterations needed until $\mathcal{R} = (0,\dots, 0)$.

An important observation is that the order in which the messages are chosen to be retransmitted does not affect the total completion time 
of the process.
Hence we can instead analyze the process where we first repeatedly set $R_1^{t+1}$ to be a random value 
in $\{0,\dots, R_1^t-1\}$ until we reach zero; and then we do the same for $R_2$ and so on until $R_k$.
This simplifies the analysis since we are now dealing with a sum of $k$ i.i.d.~random variables; and we can 
get upper and lower tail bounds for the time it takes for $R_1^t$ to reach zero by relating $R_1^t$ to the continuous
random variable $n U_1 \dots U_t$ where $U_1, \dots, U_t$ are i.i.d.~uniform on $[0,1]$.
\end{proofsketch}

\paragraph{Bounded reach scenario}

Recall that in this scenario we have $r \ll n$. 
We can prove that with high probability the actual number of messages received is only a constant factor larger than in the best case; that is, linear in $k$.

\begin{theorem}\label{thm:1d-monotone-cdp}
  For the CD-P heuristic 
  in the bounded reach scenario, assuming fair medium access,
  $\recpac \in \Theta(k)$ with high probability.
\end{theorem}
%%%%%%%%%%%%%%

\begin{proof}
Consider the message grid, where all messages start in the leftmost $r$ columns. Whenever a node retransmits a message, it gets deleted from all nodes to its left and added to the $r-1$ nodes to its right. So, each message is always present in exactly $r$ consecutive nodes (except at the end of the process), that we will refer to as a \emph{train} of $r$ messages.

In the CD-P protocol, when a node sends a message, it chooses it from the train that is most behind.

From this we can already see that the average speed at which the trains progresses towards the destination is
%smaller than $r/2$ nodes per move for CD, and 
greater than $r/2$ nodes per move for CD-P. If a train of messages does not overlap with the other trains, its expected progress during the next move is $r/2$ nodes. If a train does overlap with other trains, if moves by more than $r/2$ nodes. 

%In particular, the trains will tend to stick together more. They probably move faster on average than in CD, but it is not so easy to see how fast exactly.

%under the CD heuristic, it moves by less than $r/2$ nodes, and looking more closely, though, in CD, the caterpillars will tend to spread out, and after a while the expected number of caterpillars in a stack that we hit is $1$, which means that the average distance that is made is $r/2$ as $n \to \infty$. So, there will be approximately $2kn/r$ steps before all caterpillars made it to the end (in the beginning and at the end the behaviour is different, but for large enough $n$ this is what happens).

%In CD-P, the caterpillars will tend to stick together more. They probably move faster on average than in CD, but it is not so easy to see how fast exactly.

The expected progress for a train of $r$ messages at each move is greater than $\frac{r}{2}$. Therefore, after $\frac{2 k n}{r}$ rounds, the expected number of messages left in the queues will be $0$. 
The total number of messages sent is $O(\frac{k n}{r})$, and every node receives a fraction $O(\frac{r}{n})$ of all the messages, therefore, $\recpac=O(\frac{k n}{r}\cdot\frac{r}{n})=O(k)$.
Since $\recmess \geq k$ for any activation order, the theorem follows.
\end{proof}

\subsection{CD heuristic}

\paragraph{Unbounded reach scenario}

We can show that if $k$ is not greater than $n$, the number of messages depends only logarithmically on $n$, at the cost of a quadratic dependence on $k$.

%  For the CD-P heuristic 
%  in the bounded reach scenario, assuming fair medium access,
  %$\recpac \in \Theta(k)$ with high probability.

\begin {theorem}
\label{thm:CD-unbounded}
  For the CD heuristic in the unbounded reach scenario, assuming fair medium access, with high probability
   %(as $k$, $n$, or both tend to infinity),
%\[
%\recpac \in
%\begin{cases}
% \Theta (k^2 \log (n/k))\,, & \text{ if } k < n\,,\\
% \Theta (kn)\,, & \text{ if } k > n\,.
%\end{cases}
%\]
%
\[
\recpac \in
\begin{cases}
 \Theta (k^2 \log (\lceil n/k \rceil + 1))\,, & \text{ if } k \le n\,,\\
 \Theta (kn)\,, & \text{ if } k > n\,.
\end{cases}
\]
%
%k <= n => Theta( k^2 * log( \lceil k/n rceil + 1 )
%2) k > n => Theta( nk ).
% 
%or just
% 
%Theta( k * min(k,n) * log( \lceil k/n rceil + 1 ) ) for all k, n.
% 
\end {theorem}

\begin {proofsketch}
To prove the theorem,
we argue using the message grid, refer to Section~\ref{sec:message_grid}.
%
%model the problem as an $n\times k$ grid of squares ($n$ columns of height $k$), as shown in Fig.~\ref{fig:full_scenarios_ub}
At each round a random column is selected among those columns that are not yet completely empty. Having chosen a column, we remove the highest filled square in that column, together with all the squares to its left.
We are interested in the (random) time $T$ when all squares are empty.
Note that this is the same as asking for the time when the $n$-th column has been hit $k$ times.
What makes it tricky to analyze is that while initially we have chance $1/n$ of hitting the $n$-th column, this probability will increase as time goes on and more columns become empty.

%To prove the result for $n < k$, we observe that a column can only be empty at time $t$ if the sequence of random columns $X_1, \dots, X_t$ contains a non-increasing subsequence of length $k$, and find that a crude upper bound on $S$ is sufficient to prove the quadratic behavior.
%The case $n > k$ is more involved: it is still reasonably easy to show for $k > \log\log n$ and for constant $k$, but the transition phase requires significant probabilistic analysis.
%\end {proofsketch}

%\begin {proofsketch}
%A sketch of the proof of this theorem is as follows. 
We start with the easier case, when $k > n$.
The upper bound is immediate from the above. For the lower bound, it
is enough to consider the
time $S$ until there is some node that has received all $k$ messages.
Before time $S$, each of the $n$ nodes is equally likely to send
at every time step, which simplifies the analysis of the process. A
relatively straightforward computation involving the
Chernoff-Hoeffding bound
shows that in fact the probability that $S < c nk$ is exponentially
small, for a certain constant $c$.
Since $\text{RecMess} \geq S$, this proves the $k > n$ case.

In the case $k \leq n$, straightforward computations show that
$S = \Theta( k^2 )$ with probability close to one.
For an upper bound on $\text{RecMess}$ we consider a modified process.
In the modified process, we wait until some node with index $\geq n/2$
has received each of the $k$ messages (in which case all nodes to its
left have too), we then discard all nodes with index $\leq n/2$,  and
in the nodes with index $> n/2$ we put all $k$
messages back, and we now repeat with the smaller set of nodes. Each
of these ``phases'' is essentially governed by the earlier bound on
$S$ and after $\Theta( \log(n/k) )$ steps we are in a situation where the
number of remaining nodes is less than $k$ so that we can apply the
bounds for the other case.
The proof of the lower bound in the case  $k>n$ follows a very similar idea
but is technically more involved.
The full proof can be found in Section~\ref {appx:omitted-cd}.
\end {proofsketch}

\paragraph{Bounded reach scenario}
Next we consider the bounded reach scenario.

\begin{theorem}\label{thm:1d-monotone-cd}
  For the CD heuristic in the bounded reach scenario,   assuming fair medium access, %(as $k$, $n$, or both tend to infinity),
  $\recpac \in O(k^{3/2})$ 
  with high probability.
  %\rodrigo{This is the only one without a lower bound, do we want to leave it as it is?}
\end{theorem}

\begin{proof}

%%%%%%%%%%%
% CD starts

In CD, when a node gets to retransmit, it picks the message with the lowest id. This means that the trains of messages are always ordered by id: the lower the id, the closer the train is to the destination. Notice that the trains cannot ``overtake'' each other. Thus, when a node chooses a message to retransmit, it retransmits the message that corresponds to the train most ahead of the others. 

\begin{figure}[h]
\centering
\includegraphics{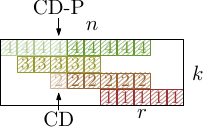}%
\caption{CD vs CD-P in bounded reach scenario.}
\end{figure}

 In particular, it moves by less than $r/2$ nodes.
  More precisely, the trains will tend to spread out, and after a while the expected number of trains in a stack that we hit is $1$, which means that the average distance that is made is $r/2$ as $n \to \infty$. So, there will be approximately $2kn/r$ steps before all trains made it to the end (in the beginning and at the end the behaviour is different, but for large enough $n$ this is what happens).

The expected progress for a train of $r$ messages is  trickier to calculate. In the case where all the trains overlap, it is
\[
\Ee(\text{progress})=\sum_1^{k-1}\frac{(R_{i+1}-R_{i})^2}{2(r+R_k-R_1)}+\frac{r^2}{2(r+R_k-R_1)}=f(R_1,R_2,\dots,R_k)\,,
\]
where $R_1\leq R_2\leq\dots\leq R_k$, and $R_i$ is the furthest node from the destination that still has message $i$, or, in other words, the tail of the train. Finding the minimum of this function, we get:
\[
\begin{split}
\frac{\partial f}{\partial R_1}&=-\frac{R_2-R_1}{r+R_k-R_1}+\sum_1^{k-1}\frac{(R_{i+1}-R_{i})^2}{2(r+R_k-R_1)^2}+\frac{r^2}{2(r+R_k-R_1)^2}=0\,,\\
\frac{\partial f}{\partial R_i}&=\frac{R_{i}-R_{i-1}}{r+R_k-R_1}-\frac{R_{i+1}-R_{i}}{r+R_k-R_1}=0\,,\text{ for $1<i<k$,}\\
\frac{\partial f}{\partial R_k}&=\frac{R_k-R_{k-1}}{r+R_k-R_1}-\sum_1^{k-1}\frac{(R_{i+1}-R_{i})^2}{2(r+R_k-R_1)^2}-\frac{r^2}{2(r+R_k-R_1)^2}=0\,.
\end{split}
\]
From the middle equations we get that $R_i=\frac{R_{i-1}+R_{i+1}}{2}$ when the expected shift is minimized. Thus, the expected shift is
\[
\Ee(\text{progress})\geq\sum_1^{k-1}\frac{(\frac{R_k-R_1}{k-1})^2}{2(r+R_k-R_1)}+\frac{r^2}{2(r+R_k-R_1)}=\frac{(R_k-R_1)^2}{2(k-1)(r+R_k-R_1)}+\frac{r^2}{2(r+R_k-R_1)}\,.
\]
Let $R_k-R_1=x$, and $\Ee(\text{progress})=f(x)$. Again find the minimum of the function:
\[
\frac{\partial f}{\partial x}=\frac{x}{(k-1)(r+x)}-\frac{x^2}{2(k-1)(r+x)^2}-\frac{r^2}{2(r+x)^2}=0\,.
\]
We get equality
\[
x^2+2x r+r^2-k r^2=0\,
\]
from which we get that the expected shift is minimized when $x=r(\sqrt{k}-1)$. Thus,
\[
\Ee(\text{progress})>\frac{r}{\sqrt{k}+1}\,.
\]

If the trains of messages do not overlap, the expected progress per move is even higher. Therefore, after $\frac{\left(\sqrt{k}+1\right) k n}{r}=O(\frac{k^{3/2}n}{r})$ steps, the expected number of messages left in the queues will be $0$.  The total number of messages sent is $O(\frac{k^{3/2}n}{r})$, and every node receives a fraction $O(\frac{r}{n})$ of all the messages, therefore, $\recpac=O(\frac{k^{3/2}n}{r}\cdot\frac{r}{n})=O(k^{3/2})$.
\end{proof}

\section{Proof of  Theorem~\ref{thm:1d-clique-cdp}: CD-P, unbounded reach}
\label{appx:omitted-cdp}
%\rodrigo{Would it make sense to describe this proof also in terms of the message grid?}
Under the CD-P heuristic, nodes prioritize messages in the queue by the distance to their senders. 
Thus, when a node $i$ is chosen to send a message, it will choose a message that is currently present in the maximum number of nodes. 
Then, all the nodes that are further away from the destination, i.e., all the nodes with index less than $i$, will delete this message from their queues.
We will show that, with high probability, it takes $\Theta(k\log n)$ steps until all nodes have emptied their queues. 

To do so we represent the number of nodes in which a message is present as random variables, and will bound them from above by 
continuous random variables with a simpler distribution at every iteration of the protocol.

We let $\mathcal{R}^{t}=(R^{t}_{1},R^{t}_{2},\dots,R^{t}_{k})$, where $R_i^t$ denotes the number of nodes that contain message $i$ at time $t$.
Thus, at time $t=0$ we have $\mathcal{R}^{0} = (n,\dots, n)$.
One round of the CD-P protocol corresponds to the largest value in $\mathcal{R}$ being randomly reduced. 
In other words,  we find an index $i$ such that $R_i^t = \max_{j=1,\dots,k} R_j^t$; we set $R_j^{t+1} := R_j^t$ for all $j\neq i$; and we choose 
$R_i^{t+1}$ uniformly at random from $\{0,\dots, R_i^t - 1\}$.

We let $T$ denote the first $t$ for which $R_1^t = \dots = R_k^t = 0$. 

\noindent
We shall always assume that $n \geq 2$, because if $n=1$ then trivially $T=n$.

\subsection{Some background on stochastic domination and coupling}

We say that $Y$ {\em stochastically dominates} $X$, denoted $X \leqst Y$, if

\[ \Pee( X > t ) \leq \Pee( Y > t ), \quad \text{ for all $t\in \eR$. } \]

A {\em coupling} of two random variables/vectors/objects $X,Y$ is a
joint probability space for $(X,Y)$ on which the marginals are correct.
That is, if $X$ has probability distribution $\Pee_X(.)$ and $Y$ has probability distribution $\Pee_Y(.)$, 
then we have a probability space $\Ccal = (\Omega, \Acal, \Pee_{\Ccal})$ where we do a single chance experiment 
that determines
both $X$ and $Y$ (formally speaking, $X,Y$ are both measurable functions from $\Omega$ to $\eR$), in such a way 
that $\Pee_\Ccal( X \in A ) = \Pee_X( X \in A)$ and $\Pee_\Ccal( Y \in A ) = \Pee_Y( Y \in A )$ for all (measurable) $A$.
For a proof of the following standard result we refer the reader to for instance~\cite{lindvallbook}.

\begin{theorem}[{Strassen's theorem}]
$X \leqst Y$ if and only if there is a coupling for $(X,Y)$ in which $X \leq Y$ with probability one. 
\end{theorem}

\subsection{The proof of Theorem~\ref{thm:1d-clique-cdp}}

We start by making the observation that the (distribution of the) completion time would be the same
if we changed the process as follows: 
first we repeatedly keep setting $R_1^{t+1}$ to a random value in $\{0,\dots R_1^t-1\}$ until we reach a situation where $R_1^t=0$; next we do the same
for $R_2$, then for $R_3$, and so on.
The reason that this gives the same distribution of the completion time is that for every time step $t$ of the original process,
every $R_i^t$ that is nonzero is sure to become the maximum value in some future round $t'\geq t$.

Let us thus write $Z$ for the (random number) of rounds it takes for $R_1^t$ to reach zero, starting from $R_1^0 = n$.
Then $T \isd Z_1+\dots+Z_k$, where $Z_1, \dots, Z_k$ are i.i.d.~distributed like $Z$.
 
Let $ U_1, U_2, \dots$ be i.i.d.~uniform on $[0,1]$ and let us denote 

\[ \ZU := \min\{ t : n \cdot U_1 \cdot \dots \cdot U_t < 1\}, \quad \ZL := \min\{ t : n \cdot U_1 \cdot \dots \cdot U_t < t+1 \}. \]

\noindent
Then we have the following relation:

 \begin{lemma}
$\ZL \leqst Z \leqst \ZU$.  
 \end{lemma}

\begin{proof}
For notational convenience we write $V_i := n U_1 \dots U_i, W_i := V_i - i$.
We start with the (stochastic) upper bound.
It is enough to show that $R_1^i \leqst V_i$ for all $i$, since that will prove that 

\begin{equation}\label{eq:romanesco} 
\Pee( Z > t ) = \Pee(  R_1^{\lfloor t\rfloor} \geq 1 ) \leq \Pee( V_{\lfloor t\rfloor} \geq 1 ) = \Pee( \ZU > t ). 
\end{equation}

We will use induction to prove $R_1^i \leqst V_i$ for all $i$.
Clearly, $R_1^0 = n = V_0$ so the base case holds.
Let us thus suppose that $R_1^i \leqst V_i$ for some $i$. By Strassen's theorem, there is 
a coupling for $(R_1^i, V_i)$ so that $R_1^i \leq V_i$ with probability one.
If under this coupling, we condition on the event that $V_i=x, R_1^i = \ell$ (where necessarily $\ell \leq \lfloor x \rfloor$), then 
$V_{i+1} = V_i U_{i+1}$ is chosen uniformly 
at random on the interval $[0,x]$, while $R_1^{i+1}$ is chosen uniformly from the set
$\{0,\dots, \ell-1\}$.
Thus, for all $0 \leq y \leq x$ we have 

\[ \Pee( V_{i+1} \geq y | V_i = x, R_1^i = \ell ) = \frac{x-y}{x} = 1 - \frac{y}{x}, \]

\noindent 
and 

\[
 \Pee( R_1^{i+1} \geq y | V_i = x, R_1^i = \ell ) = \frac{(\ell-1) - \lceil y \rceil}{\ell} = 
 1 - \frac{\lceil y\rceil + 1}{\ell} \leq 1 - \frac{y}{x}.
\]

\noindent
Summing over $\ell$, this simplifies to 

\begin{equation}\label{eq:snotje1} 
\Pee( V_{i+1} \geq y | V_i = x ) = \sum_\ell \Pee( V_{i+1} \geq y | V_i = x, R_1^i = \ell ) \Pee( R_1^i = \ell | V_i = x ) = 1 - \frac{y}{x},  
\end{equation}

\noindent
and similarly, 

\begin{equation}\label{eq:snotje2} 
\Pee( R_1^{i+1} \geq y | V_i = x ) \leq 1-\frac{y}{x}.  
\end{equation}

\noindent
Hence 

\begin{equation}\label{eq:snotje3} 
\begin{array}{rcl} 
\Pee( V_{i+1} \geq y ) 
& = & \int_y^n \Pee( V_{i+1} \geq y | V_i = x ) f_{V_i}(x){\dd} x \\
& \geq & \int_y^n \Pee( R_1^{i+1} \geq y | V_i = x ) f_{V_i}(x){\dd} x \\
& = & \Pee( R_1^{i+1} \geq y ), 
   \end{array}
\end{equation}

\noindent
so that the upper bound is proved.

For the lower bound, we remark it is sufficient to show $W_i \leqst R_1^i$ for all $i$, by an argument 
analogous to~\eqref{eq:romanesco}.
Again we use induction, and again the base case is trivial, as $R_1^0 = n = W_0$.
For the induction step, we can assume there exists a coupling such that 
$W_i \leq R_1^i$ with probability one, under this coupling.
If we condition on the event $W_i = x, E_1^i = \ell$ (where now necessarily $x \leq \ell$), then
$W_{i+1} = U_{i+1} (W_i+i) - (i+1)$ is uniform on the interval $[-(i+1), x-1]$ while
$R_1^i$ is uniform on $\{0,\dots, \ell-1\}$.
Hence for $0 \leq y \leq x-1$ we have that

\[ \Pee( W_{i+1}\geq y | W_i = x, R_1^i = \ell ) = \frac{x-1-y}{x+i} = 1 - \frac{y+i+1}{x+i}, \]

\noindent
while 

\[ 
 \Pee( R_1^{i+1} \geq y | W_i = x, R_1^i = \ell ) = \frac{\ell-1-\lceil y\rceil}{\ell} = 1 - \frac{\lceil y\rceil + 1}{\ell}
 \geq 1 - \frac{y+i+1}{x+i},
\]

\noindent
using that $\frac{\lceil y\rceil + 1}{\ell} \leq \frac{y+1}{x} \leq \frac{y+1+i}{x+i}$
(repeatedly applying that $a>b \geq 0$ implies that $\frac{b}{a} \leq \frac{b+1}{a+1}$).
Thus, arguing completely analogously to~\eqref{eq:snotje1},~\eqref{eq:snotje2} and~\eqref{eq:snotje3}, the lower bound is proved.
(Here we should remark it is sufficient to consider only $y \geq 0$, since $\Pee( R_1^{i+1} \geq y ) = 1$ for $y < 0$.)
\end{proof}

\begin{corollary}\label{cor:Z1UB}
$\Pee( Z_1 > t ) \leq n (1/2)^t$.
\end{corollary}

\begin{proof}
This follows by Markov's inequality, since $\Ee\left( n U_1\dots U_t \right)=n(1/2)^t$.
\end{proof}

\begin{corollary}\label{cor:Z1LB}
There exists a universal constant $c > 0$ such that

\[ \Pee\left( Z < c \log n \right) \leq \frac{1}{\log n} \quad  \text{(for all $n\geq 2$)} \]

\noindent
(Note that the distribution of $Z$ implicitly depends on $n$.)
\end{corollary}

The proof of this corollary makes use of the observation that if $U$ is uniform on $[0,1]$ and $W := \ln(1/U)$ then
$\Ee W = \Var W = 1$.
For completeness we spell out the straightforward computations.

\begin{lemma}\label{lem:EW}
Let $U$ be uniform on $[0,1]$ and let $W := \ln(1/U)$.
Then $\Ee W = \Var W = 1$.
\end{lemma}

\begin{proof}
We have that 

\[ \Ee W = \int_0^1 - \ln u {\dd}u = \left[ -u\ln u + u \right]_0^1 = 1, \]

\noindent
and 

\[ \Ee W^2 = \int_0^1 \ln^2 u {\dd} u = \left[ u \ln^2 u - 2u\ln u + 2u \right]_0^1 = 2. \]

\noindent
So indeed $\Var W = \Ee W^2 - (\Ee W)^2 = 1$.
\end{proof}

\begin{proofof}{Corollary~\ref{cor:Z1LB}}
Since $\Pee(Z = 0 ) = 0$, we can assume that $n \geq n_0$ for some sufficiently large constant $n_0$ (at the price of 
possibly having to lower the universal constant $c$ a bit).
We put $t := \lfloor \frac{\ln n}{1000}\rfloor$ and we  note that 

\[ \Pee\left[ n U_1 \dots U_t < t+1 \right]
 = \Pee\left[ \ln(1/U_1) + \dots + \ln(1/U_t) > \ln n - \ln(t+1) \right].
\]

\noindent
Writing $S := \ln(1/U_1) + \dots + \ln(1/U_t)$ for convenience, we see that
$\Ee S = \Var S = t$ by Lemma~\ref{lem:EW} above, and hence

\[ \begin{array}{rcl} 
\Pee\left[ S > \ln n - \ln(t+1) \right]
& = & \Pee\left[ S - \Ee S > \ln n - \ln(t+1) - \Ee S \right]  \\
& \leq & \Pee\left[ S - \Ee S > \ln n - t - \Ee S \right]  \\
& \leq & \frac{t}{(\ln n - 2t)^2} \\
& \leq & \frac{1}{\ln n},
\end{array} \]
 
\noindent
having used $t \geq \ln(t+1)$ in the second line; Chebyschev in the third line; and that $n$ is sufficiently large in the fourth line.
\end{proofof}

\begin{corollary}\label{cor:mainolog}
There exist universal constants $c, C > 0$ such that, if $n \to \infty$ and $k=k(n)$ satisfies $1 \leq k = o(\log n)$ then
\[ \Pee( c k \log n < T < C k \log n ) \to 1. \]
\end{corollary}

\begin{proof}
For the upper bound, we notice that

\[ \Pee( T > C k\log n ) \leq k \Pee( Z > C \log n ) \leq kn (1/2)^{C\log n} \to 0, \]

\noindent
using Corollary~\ref{cor:Z1UB} and where the limit follows assuming $C$ is sufficiently large.
For the lower bound, we argue similarly noting that

\[ \Pee( T < c k\log n ) \leq k \Pee( Z_1 < c \log n ) \leq \frac{k}{\log n} \to 0, \]

\noindent
using Corollary~\ref{cor:Z1LB} (assuming $c$ was chosen smaller than the constant provided there) and using that $k=o(\log n)$.
\end{proof}

This last corollary proves our main result in the case when $n\to\infty$ and $k$ is either constant or does not go to infinity
too fast. To complete the proof of the main result, it now suffices to consider the case when $k\to\infty$ and 
$n = n(k) \geq 2$ is arbitrary.
Before we tackle this case we need one more observation.

\begin{corollary}\label{cor:EZVarZ}
There exist universal constants $c, C > 0$ such that 

\[ c \log n \leq \Ee Z \leq C \log n, \quad c \log^2 n \leq \Ee Z^2 \leq C \log^2 n, \quad \text{(for all $n \geq 2$)} \]

\noindent
(Again we remind the reader that $Z$ depends implicitly on $n$.)
\end{corollary}

\begin{proof}
It follows from Corollary~\ref{cor:Z1LB} that

\[ \Ee Z \geq c \log n \cdot \Pee( Z \geq c \log n ) = \Omega( \log n ), \]

\noindent
if $c$ is as in Corollary~\ref{cor:Z1LB}.
The same argument also gives $\Ee Z^2 \geq \left( \Ee Z \right)^2 = \Omega( \log^2 n )$.

By a standard formula for the expectation of nonnegative random variables, we have

\[ \begin{array}{rcl}
    \Ee Z & = & \int_0^\infty \Pee( Z \geq t ) {\dd} t \\
    & \leq & C \log n + \int_{C\log n}^\infty n (1/2)^t {\dd} t \\
    & = & O(\log n), 
   \end{array} \]

   \noindent
provided $C$ is chosen sufficiently large.
Similarly

\[ \begin{array}{rcl}
    \Ee Z^2 & = & \int_0^\infty \Pee( Z \geq \sqrt{t} ) {\dd} t \\
    & \leq & C \log^2 n + \int_{C\log^2 n}^\infty n (1/2)^{\sqrt{t}} {\dd} t \\
    & = & O(\log^2  n) + n \int_{C \log^2 n}^\infty  e^{\ln(2) \sqrt{t}} {\dd} t\\
    & = & O(\log^2  n) + n O\left( \int_{\ln 2 \sqrt{C} \log n}^\infty z e^{-z} {\dd} z \right) \\
    & = & O(\log^2  n) + n O\left( \int_{\ln 2 \sqrt{C} \log n}^\infty e^{-z/2} {\dd} z \right) \\
    & = & O(\log^2  n),
   \end{array} \]
   
\noindent
using the substitution $z := \ln(2) \sqrt{t}$ in the third line, and that $z \leq e^{z/2}$
for all sufficiently large $z$ in the third line.
\end{proof}

\begin{corollary}\label{cor:ktoinf}
There exist universal constants $c, C$ such that, if $k\to \infty$ and $n=n(k) \geq 2$ is arbitrary then
\[ \Pee( c k \log n \leq T \leq C k \log n ) \to 1. \]
\end{corollary}

\begin{proof}
Since $T = Z_1 + \dots + Z_k$, the previous corollary shows that 
$\Ee T = k \Ee Z = \Theta( k \log n )$ and $\Var T = k \Var Z = \Theta( k \log^2 n )$.
It follows that 

\[ \Pee( |T - \Ee T| > \frac12 \Ee T ) \leq \frac{4 \Var T}{(\Ee T)^2} = O(\frac{1}{k} ) = o(1), \]

\noindent
using Chebyschev and Corollary~\ref{cor:EZVarZ}.
To conclude, we observe that, provided we chose $c$ small enough and $C$ large enough, we have
$\{T < c k \log n \}, \{ T > C k \log n \} \subseteq \{ |T-\Ee T| > \frac12 \Ee T \}$. 
\end{proof}

\begin{proofof}{Theorem~\ref{thm:1d-clique-cdp}}
We remark that Corollaries~\ref{cor:mainolog} and~\ref{cor:ktoinf} together prove the theorem, taking the
minimum of the two $c$'s and the maximum of the to $C$s provided by these corollaries for the full result.
\end{proofof}

\section{Proof of Theorem~\ref{thm:CD-unbounded}: CD, unbounded reach}
\label{appx:omitted-cd}
We model the problem with the message grid, an $n\times k$ grid of squares ($n$ columns of height $k$).
In this section, for ease of exposition we assume that the grid is initially empty, and that we are interested in the number of time steps (messages sent) before the grid is completely full.
That is,
there are discrete time steps $t=1,2, \dots$ and initially the grid is empty.
At each time step a random column is selected among those columns that are not yet completely full; the column is chosen uniformly at random 
among all the non-full columns.
Having chosen a column, we fill the lowest square in the column that is still not filled, together with all the squares to the left of 
it (insofar as those are not yet filled).
So if in the first time step column $i$ was selected, then squares $(1,1), \dots, (i,1)$ get filled.
If in the second time step column $j$ was selected then there are two possibilities. If $j \leq i$ then squares $(2,1), \dots (2,j)$ get filled. 
Otherwise, squares $(1, i+1), \dots, (1, j)$ get filled.\footnote {Clearly, this formulation is equivalent to starting with a full grid, and removing squares (messages) when a random non-empty column gets chosen.}

We are interested in the (random) time $T$ when all squares have become filled.
Note that this is the same as asking for the time when the $n$-th column has been hit $k$ times.
What makes it a little bit tricky to analyze is that while initially we have a  chance of $1/n$ of hitting the $n$-th column, this probability will 
increase as time goes on and more columns get filled.
Below we consider the case where the number of columns $n$ or the number of rows $k$ goes to infinity (or both).

%\subsection*{The result}

In order to prove Theorem~\ref{thm:CD-unbounded}, we prove the following, more specific, result.

\begin{theorem}\label{thm:main}
There exist constants $0 < c_L < c_U < \infty$ such that:
\begin{enumerate}
\item\label{itm:ngeqk} If $k = k(n) \leq n/1000$ then 
 \[ \Pee\left( c_L \cdot k^2 \log\left(\frac{n}{k}\right) \leq T \leq c_U \cdot k^2 \log\left(\frac{n}{k}\right) \right) \to 1, 
 \]
\noindent
as $n \to \infty$;
\item\label{itm:kgeqn} If $k = k(n) \geq n/1000$ then 
 \[ \Pee\left( c_L \cdot kn \leq T \leq c_U \cdot kn \right) \to 1, 
 \]
 \noindent
as $n \to \infty$.
\item\label{itm:nconst} If $n$ is constant then
\[ \Pee\left( c_L \cdot kn \leq T \leq c_U \cdot kn \right) \to 1, 
 \]
 \noindent
as $k \to \infty$.
\end{enumerate}
\end{theorem}

In the following sections we prove Theorem~\ref{thm:main}.

\subsection{The random variable $S$}

We will consider the time $S$ until at least one column becomes full.
Of course, $S \leq T$. What simplifies the analysis of $S$ as opposed to $T$ is that until time $S$ in each time step we are choosing a column
uniformly at random from all $n$ columns.

\begin{lemma}\label{lem:SngeqkLower}
For all $k \leq n$ and all $t$ we have that:
\[ \Pee( S < t ) \leq \left(\frac{2e^2t}{k^2}\right)^k.\]
\end{lemma}

\begin{proof}
Let $X_1, X_2, \dots$ denote the indices of the columns selected in each time step.
Since we are only interested in $S$, we can just pretend that nothing changes in the way we choose columns
after one or more have filled up.
In other words, from now on in this proof we suppose we are given an infinite sequence $X_1, X_2, \dots$ of random variables (column indices) 
that are chosen independently and uniformly at random from $\{1,\dots, n\}$.

We now make the following crucial observation: 
There is a full column at time $t$ if and only if there is a sequence
$t_1 < \dots < t_k$ of length $k$ such that $t_k \leq t$ and 
$X_{t_1} \geq X_{t_2} \geq \dots \geq X_{t_k}$.
That is, the sequence $X_1, \dots, X_t$ must contain a non-increasing subsequence of length $k$.\footnote {This is loosely connected to the famous longest increasing subsequence problem for random permutations.
Also there we need a permutation of length $\text{const} \times n^2$ before we can expect to find an increasing subsequence of length
$n$.}
To complete the proof, we are going to give a very crude upper bound on the probability that there exists such a non-increasing subsequence 
in $X_1, \dots X_{t}$.
We will use the first moment method, i.e.,~we will just count---or rather bound---the expected number of all such subsequences.

First we choose the $k$ time steps $t_1, \dots, t_k$ on which a non-increasing subsequence is to occur. 
This can be done in 

\[ {t \choose k} \leq \left(\frac{et}{k}\right)^k,  \] %= (ec)^{k}, \]

\noindent
ways (using a standard bound on binomial coefficients).

Assuming that we know the time steps $t_1, \dots, t_k$ on which a non-increasing subsequence is to be formed, we wish to 
count the number of ways we can choose values for $X_{t_1},\dots, X_{t_k}$ to make them non-increasing.
To specify a non-increasing sequence of length $k$ with elements $\in \{1,\dots,n\}$ it is enough to specify how many 
times each number is hit.
That is, we have $a_1,\dots, a_n \in \{0,\dots, k\}$ such that $a_1+\dots+a_n=k$, where $a_i$ is the number of times we use number $i$
in the sequence.
The sequence will start with $a_n$ times the number $n$, followed by $a_{n-1}$ times the number $n-1$, and so on.
The number of ways of choosing values for $X_{t_1},\dots, X_{t_k}$ is thus:

\[ {k + n-1 \choose k} \leq \left(\frac{e(n+k)}{k}\right)^k \leq \left(\frac{2en}{k}\right)^k. 
\]

Next, we remark that for any given $t_1, \dots, t_k$ and $a_1, \dots, a_k$ the probability that 
$X_{t_1}, \dots, X_{t_k}$ take the values specified (implicitly) by $a_1,\dots, a_n$ simply equals
$n^{-k}$.

Putting it all together we find that 

\[ \Pee( S < t ) \leq (et/k)^k \cdot (2en/k)^k \cdot n^{-k}  
 = (2e^2t/k^2)^k,
\]

\noindent
as required.
\end{proof}

\begin{lemma}\label{lem:Skgeqn}
There exists a universal constant $c > 0$ such that
for all $k \geq n$ and all $t \leq cnk$,
\[ \Pee( S < t ) \leq \left(\frac{16t}{nk}\right)^{k/2}. \]
\end{lemma}

\begin{proof}
The argument is similar to that in the proof of Lemma~\ref{lem:SngeqkLower}, except that now we do not choose the times $t_1, \dots, t_k$.
Instead we observe that if there exists a non-increasing subsequence that corresponds
to $a_1, \dots, a_n$ (as above, with $a_1+\dots+a_n=k$, and so) then 
one can construct (possibly another) such sequence by first waiting
until the time $t_n$ when we have first seen $a_n$ occurrences of $n$, then waiting
until the first time $t_{n-1} > t_n$ such that we have seen $a_{n-1}$ occurrences of $n-1$ 
in the interval $(t_{n}, t_{n-1}]$, and so on.
The sequence we obtain via this procedure will be the first to occur among all sequences corresponding to the
same vector $\overline{a} = (a_1,\dots, a_n)$.
For a given such vector $\overline{a}$ we can therefore compute the probability 
that a subsequence corresponding to $\overline{a}$ occurs before time $t$ as follows:

\begin{equation}\label{eq:BiAap} 
\begin{array}{c}
\Pee( \text{there exists a subsequence corresponding to $\overline{a}$ in $X_1, \dots, X_t$} ) \\
 = \\
\Pee\left( Y \leq t \right),
\end{array}
 \end{equation}

\noindent
where $Y = Y_1+\dots+Y_k$ is a sum of i.i.d.~geometrically distributed random variables with common success probability
$p = 1/n$.
To see this, one can think of first ``listening to the $n$-th column'' until it has been selected $a_n$ times, then
``listening to the $(n-1)$-st column'', and so on. At each time step the desired column is selected with probability $1/n$, 
independently of the previous history of the process.

By Lemma~\ref{lem:chernoffsumgeo} below, we have 

\[ 
\Pee\left( Y \leq t \right) 
\leq
\exp\left[ - (kt / \Ee (Y)) \cdot H(\Ee (Y)/t) \right] 
\leq
\exp\left[ - (t/n) \cdot H(nk/t) \right]. %< e^{-3k }. 
\]

\noindent
(Here $H(x) := x\ln x - x + 1$, and we have used $\Ee (Y) = nk$.)
Now notice that $H(x) \geq \frac12 \cdot x\ln x$ for all $x$ sufficiently large.
That is, for all $x\geq x_0$ with $x_0$ a suitable constant---for instance, taking $x_0 := e^2$ is sufficient.
Thus, we can choose $c$ such that 
$(t/n) \cdot H(nk/t) \geq (k/2) \ln( nk/t )$ for all $t \leq c nk$.

The number of sequences $\overline{a}$ is at most

\[ {n+k-1\choose k} \leq {2k \choose k} \leq 4^k. \]

\noindent
Putting these bounds together we find that, for $t \leq cnk$: 

\[ \Pee( S < t ) \leq 4^k \left(\frac{t}{nk}\right)^{k/2} = \left(\frac{16t}{nk}\right)^{k/2}, 
\]

\noindent 
as required.
\end{proof}

Combining Lemmas~\ref{lem:SngeqkLower} and~\ref{lem:Skgeqn} we see that:

\begin{corollary}\label{cor:Skgegngedeeld100}
There exists a universal constant $c>0$ such that 
$\Pee( S < c nk ) \to 0$ as $n \to \infty$ for any sequence $k=k(n)$ satisfying $k \geq n/100$.
\end{corollary}

\begin{lemma}\label{lem:SngeqkUpper}
There exists an absolute constant $C > 0$ such that for all $n,k$ and all $t \geq Ck^2$ we have
 
\[ \Pee( S > t ) \leq e^{-t }. \]

\end{lemma}

\begin{proof} (sketch)
We can assume without loss of generality that $k < n/1000$
(otherwise, if $k \geq n/1000$ then $S \leq 1000 k^2$ is deterministically true, 
since in every time step a square gets filled and there are 
$nk \leq 1000 k^2$ squares in total).
First we wait until the first time one of the last $\lfloor n/k\rfloor$ columns are hit, then we wait
until the next time after that the next $\lfloor n/k\rfloor$ columns are hit, and so on.
%Assuming no column gets filled until times $t$, 
The probability that at time time $t$, we have not had $k$ successes in the 
above process equals the probability that $X > t$ where $X$ is the sum of $k$ independent $\geo(p)$ random variables with 
$p = \lfloor n/k\rfloor / n = \Theta( 1/k )$. We can thus apply Lemma~\ref{lem:chernoffsumgeo}.
Note that $\Ee (X) = \Theta( k^2 )$ and that $H(x) = x\ln x - x + 1 \geq x$ for all $x \geq x_0$ (where $x_0$ is of course
a suitable constant, $x_0 := e^2$ will do).
Thus, we can indeed choose $C > 0$ such that $t / \Ee (X) \geq x_0$ for all $t \geq C k^2$, giving 

\[ \Pee( S > t ) \leq \Pee( X > t ) \leq e^{- \Ee (X) \cdot H(t/\Ee (X)) } \leq e^{-t}, \]

\noindent 
for all $k,n,t$ satisfying the conditions of the lemma.
\end{proof}

\subsection{Proof of parts~\ref{itm:kgeqn} and~\ref{itm:nconst} of Theorem~\ref{thm:main}}

Note that in each step of the process at least one square of the grid gets filled.
Thus $T \leq nk$ (with probability one).
Since $T \geq S$ the lower bound in~\ref{itm:kgeqn} follows immediately from Corollary~\ref{cor:Skgegngedeeld100} 
and the lower bound in~\ref{itm:nconst} from Lemma~\ref{lem:SngeqkLower}.

\subsection{Proof of the upper bound of part~\ref{itm:ngeqk} of Theorem~\ref{thm:main}}

We will consider a modification of the process that will certainly take at least as long as the original version, and is easier to analyze.
We distinguish different phases of the process.
In the first phase, starting from the empty $k\times n$ message grid, we throw away columns until at least one column with index $\geq n/2$ is full.
While doing this, we ignore the height constraint on columns with index $\leq n/2$.
That is, columns with index $\leq n/2$ are allowed to have height $\geq k$ and we will still select columns uniformly from all $n$ columns.
As soon as a column with index $> n/2$ obtains height $k$, the first phase ends.
We now throw away the left half of the grid, and empty the right half.
So we now have an empty $(k/2)\times n$ grid. We now repeat the modified process, i.e.,~randomly select columns and
ignore the height constraints of the left half until we get some full column in the right half
(which corresponds to the rightmost quarter of the original grid) and repeat.
%all squares in the right half of the $k\times n$ grid, and we start throwing columns into the right half
%until we first get a full column with index $> \frac34 n$ and while doing so we ingore the height constraint
%for columns of index $\leq \frac34 n$.
%Once we have a full column of index $> \frac34 n$, the second phase ends, we throw away the contents of the last $\frac14 n$ columns, and so on.
We keep going on until we start a phase starts where the number of remaining columns is $\leq 1000 k$.
In the last phase we simply carry out the process as usual, waiting until all the last column is filled.
(Observe that there will be $\Theta( \log(n/k) )$ phases in total.)

\vspace{12pt}

\begin{proofof}{the upper bound}
Let $C$ as provided by Lemma~\ref{lem:SngeqkUpper}. 
Notice that if we restrict attention to the right half of the columns that have not yet been filled at the start of the phase
(i.e.,~in phase $i$ these are columns with indices $n 2^{-(i+1)} + 1, \dots, n$) then the number of times this set of columns is hit 
in a period of $t$ time steps since the start of the phase behaves like a $\Bi(t,\frac12)$ random variable (lets call it $X$) and
the probability that the phase lasts longer than $t$ is precisely the probability that 
$S_{k,n 2^{-(i+1)}} > X$.
We see that for all $i \leq I$
(where $I := \lceil\log(n/1000 k)\rceil$ is the smallest $i$ such that $n 2^{-i} \leq 1000 k$)
and all $t \geq 100 C k^2$:

\begin{equation}\label{eq:poepje} 
\begin{array}{rcl} 
\Pee(\text{phase $i$ takes longer that $t$} )
& \leq & 
\Pee( \Bi( t, \frac12 ) < t/100 ) + \Pee( S_{k,n 2^{-(i+1)}} > t/100 ) \\
& \leq & e^{-(t/2)\cdot H(2/100)} + e^{-t/100} \\
& \leq & 2 e^{-t/100},
\end{array} 
\end{equation}

\noindent
where we have applied the vanilla Chernoff bound,  Lemma~\ref{lem:SngeqkUpper}, 
and that $\frac12 H(2/100) > 1/100$.

We distinguish two cases. First, we suppose that $k \geq \sqrt{n}$.
In that case~\eqref{eq:poepje} gives that  

\[ 
\Pee(\text{phase $i$ takes longer that $100 C k^2$} ) = \exp[ - \Omega(n) ]. 
\]

\noindent
Hence, we see that if we write $K := 100 C k^2 I + 1000 k^2$ ($=\Theta(k^2\log(n/k)$), then
assuming $k\geq\sqrt{n}$ we have

\[ \Pee( \text{the process takes longer than $K$ to complete} )
 \leq \log n \cdot e^{-\Omega(n)} = o(1).
\]

\noindent
It remains to consider the case when $k \leq \sqrt{n}$. 
Note that in this case we have $I \to\infty$.
Let $L_i$ denote the duration of phase $i$.
By Lemma~\ref{lem:SngeqkLower} we have that

\begin{equation}\label{eq:eend1}
\Ee (L_i) \geq \frac{k^2}{4e^2} \cdot \Pee( L_i \geq \frac{k^2}{4e^2} ) 
 \geq \frac{k^2}{4e^2} \cdot \frac12 = \Omega( k^2 ).
\end{equation}

On the other hand, applying Lemma~\ref{lem:SngeqkUpper} and a standard formula for the expectation of non-negative random variables 
we have that 

\begin{equation}\label{eq:eend2} \begin{array}{rcl}
 \Ee (L_i) & = & %\int_0^\infty \Pee( L_i > t ) {\dd}t \\
%& = &  
 \int_0^\infty \Pee( L_i > t ) {\dd}t \\
 & \leq & 
 100 C k^2 + 2 \int_{100 C k^2}^\infty e^{-t/100}{\dd t} \\
 & \leq & 
 100 C k^2 + 2 \int_0^\infty e^{-t/100}{\dd t} \\
 & = & 
 O( k^2 ).
\end{array} \end{equation}

\noindent
(Here it is important to note that the constants hidden inside the $\Omega, O$ notation in~\eqref{eq:eend1} and~\eqref{eq:eend2}
are universal. That is, the same constants work for all $n,k$ in the range considered.)
Similarly to~\eqref{eq:eend2} we also have that:

\[ \begin{array}{rcl}
 \Ee (L_i^2) & = & \int_0^\infty \Pee( L_i^2 > t ) {\dd}t \\
& = &  
 \int_0^\infty \Pee( L_i > \sqrt{t} ) {\dd}t \\
 & \leq & 
 10^4 C^2 k^4 + 2 \int_{10^4 C^2 k^4} e^{-\sqrt{t}/10}{\dd t} \\
 & = & 
 O( k^2 ),
\end{array} \]

\noindent
where the multiplicative constant hidden inside the $O(k^4)$ is universal.

Writing $L := \sum_{i=1}^I L_i$, we have that $\Ee (L) = \Theta( I k^2 )$ and $\Var L = O( I k^4 )$.
Thus, we can apply Chebyschev's inequality to see that

\[ \Pee( L \leq \Ee (L) / 2 ) %\leq \Pee( |L-\Ee (L)| \geq \Ee (L) /2 ) 
\leq \frac{4 \Var L}{(\Ee (L))^2} = O( 1/I ) = o(1), \]

\noindent
Since $I = \Theta( \log(n/k) )$ tends to infinity under the assumption that $k\leq \sqrt{n}$.
As $L$ is the time until the process completes, this completes the proof of the upper bound.
\end{proofof}

\subsection{Proof of the lower bound of part~\ref{itm:ngeqk} of Theorem~\ref{thm:main}}

%\TM{ Dit is een stuk ingewikkelder geworden dan je (in elk geval ikzelf) in eerste instantie misschien zou denken. Er is makkelijk bewijs dat werkt voor alle $k(n) = \Omega( \log\log n )$. Dan is de grens in Lemmas~\ref{lem:SngeqkLower} namelijk nog sterk genoeg om te zorgen dat je bij $\log(n)$ keer wachten op een non-increasing sequence het telkens niet sneller dan $\Omega(k^2)$ gaat. Als $k$ constant is valt ook een makkelijk bewijs te geven. Juist als $1 \ll k(n) \ll \log\log n$ lijkt het lastiger te worden. Om het wat the streamlinen heb ik overigens met een splitsing in twee gevallen $k \leq \sqrt{n}$, resp.~$k>\sqrt{n}$, gewerkt.}

%\medskip

Again we will modify the process, this time so as to obtain a process 
that will terminate quicker (and is easier to analyze).
We will choose constants $0 < \alpha, \beta, \gamma, \delta, \eps < 1$ satisfying some
additional demands (which include $\alpha \ll \beta, \beta < \gamma^2$) that will be specified as we go along with the proof.

\begin{itemize}
 \item We divide time into ``rounds'' (periods) of length $t = \delta k^2$. 
 %(So in a round we drop $\Po(t)$ columns).
 \item We divide the columns into groups. Group $i$ consists of $n_i := n \alpha^{i}$ consecutive columns.
 (Group 1 contains columns $1, \dots, n_1$, group two
 columns $n_1+ 1, \dots n_1+n_2$ and so on.) 
 We shall be choosing $\alpha$ small enough so that in fact $\sum n_i < n$. The fact that the sum is strictly less than $n$ 
 will not be a problem for our argument.\footnote{We are ignoring rounding for now. This can be settled by noting we can assume $n, k$ and $\alpha, \beta, \gamma, \delta$ are all powers
 of two (for $n,k$ this can be shown using the monotonicity of $T$ in $n,k$ -- worsening the constant $c_L$ a bit.)}
 \item At the start of each round, the columns of groups $1, \dots, i$ will already have been ``filled'', while
 none of the columns in groups $i+1, \dots$ will be full. (The first round is round 0.)
 At the end of the period we will fill one or more additional groups of columns completely, according to the following rules:
 \begin{itemize}
 \item If at the start of the current round group $i$ is the last filled group as above, and during the round in some group $i+j$ a 
 non-increasing sequence of length $k \gamma^{j}$ gets created {\bf only taking into account what happened in the current period to the columns
 in that group}, 
 then we fill all columns up to and including that entire group, and move on to the next period. 
 (We do not drop any more columns during the rest of the round. That is, as soon as a non-increasing sequence of length $k \gamma^{j}$ gets created
 in group $i+j$, the round ends.)
 \item Otherwise, if by the end of a round no group obtained a non-increasing sequence of the desired length
 (among the selections in that period), then we
 fill up group $i$, and move on to the next period.
 \end{itemize}
 \item The process terminates once a group with index $\geq I := \eps \log(n/k)$ has been filled.
\end{itemize}

\paragraph{Why this works.} 
Note that in a round when groups $1, \dots, i$ have been filled,
but group $i+1$ has not,   the longest non-increasing sequence
that could possibly exist  in group $i+j$ (now taking into account all periods) 
is no longer than 

\[ \begin{array}{rcl}
   \sum_{0 \leq k \leq i} k \gamma^{i+j-k} 
& = & \sum_{\ell=j}^{i+j-1} k \gamma^\ell  \\
& \leq & k \gamma^{j} / (1-\gamma).
   \end{array}
\]

\noindent
Thus, even if we are very lucky and all the longest non-increasing sequences of the groups $i+1, i+2, \dots $ can be combined into
one non-increasing sequence, then this sequence would still have length at most 

\[ \sum_{j=1}^\infty k \gamma^{j} / (1-\gamma)
 = \frac{\gamma k}{(1-\gamma)^2}.
\]

This also implies that even if in the current round
some group achieves the goal we have set for it (namely a sequence of length $k \gamma^{j}$ 
for group $i+j$), then the longest sequence will still be no longer than $k \gamma (1 + 1/(1-\gamma)^2) 
< k$ (provided we chose the constants $\gamma$ sufficiently small -- which we can assume without loss of generality).
In particular, we never fill columns completely, except at the end of a round when we fill one or more entire groups of columns.

Next, we need a bound on the probability that group $i+j$ obtains a 
 non-increasing sequence of length $k \gamma^{j}$ in the period when groups $1, \dots, i$ 
have been filled but group $i+1$ has not.

\begin{lemma}\label{lem:10}
Provided $\alpha, \beta, \gamma, \delta, \eps$ are chosen appropriately, and 
$i,j$ are such that $i+j \leq I$ (the number of groups),
%$k \gamma^i \leq n \alpha^{i}$
%(in other words $i \leq \log(n/k) / \log(\gamma/\alpha) 
%= \Theta( \log(n/k) )$) 
the following holds. 
Consider the situation where we start a round in which groups $1, \dots, i$ have been filled and 
groups $i+1, i+2, \dots$ have not.
Let $E_{i,j}$ be the event that group $i+j$ achieves a sequence of length $k \gamma^j$ during this period.
Then 
\[ \Pee( E_{i,j} ) \leq 10^{-j}. \]
\noindent

Moreover, if $k \geq \sqrt{n}$, then we even have 
\[ 
 \Pee( E_{i,j} ) = \exp[-\Omega(n^{1/4})],
\]
\noindent
for all $1\leq i \leq I-1$ and $1 \leq j\leq I-i$.
\end{lemma}

\begin{proof}
Note that if $E_{i,j}$ holds,  then one of the following two possibilities must have occurred.
Either {\bf(a)} group $j$ has received $\geq k^2 \beta^j$ columns (in the current period), or
{\bf(b)} it has received less than $k^2 \beta^j$ columns,  but nonetheless a sequence of length
$k \gamma^j$ was created.
By Lemma~\ref{lem:SngeqkLower} (applied with $k' = \max(k \gamma^j,1), t' := k^2 \beta^j, n' := 
n \alpha^{j}$) we see that

\begin{equation}\label{eq:b} 
\begin{array}{rcl}
\Pee( \text{\bf(b)} ) 
& \leq & \left( \frac{2 e^2 t'}{k'^2}\right)^{k'} \\
& = & \left( \frac{2 e^2\beta^j}{\gamma^{2j}} \right)^{\max(k \gamma^j,1)} \\
& \leq & 2 e^2 (\beta/\gamma^2)^j.
\end{array}
\end{equation}

\noindent
To deal with the probability that {\bf(a)} holds we distinguish two further possibilities:
{\bf(a-1)} we have $k^2 \beta^j \geq 1$, or {\bf(a-2)} we have 
$k^2 \beta^j < 1$.
We find that under assumption {\bf (a-1)} we can apply the vanilla Chernoff bound to 
obtain

\begin{equation}\label{eq:a1} 
\begin{array}{rcl} 
\Pee( \text{\bf(a)} )
& = & \Pee( \Bi( \delta k^2, \alpha^j ) \geq k^2 \beta^j ) \\
& \leq &\exp[ - \delta k^2 \alpha^j H( (\beta/\alpha)^j ) ] \\
& \leq & \exp[ - \delta k^2 \beta^j \ln( (\beta/\alpha)^j ) / 2 ] \\
& \leq & (\alpha/\beta)^{j/2}.
\end{array} 
\end{equation}

\noindent 
where in the third line we have used that $H(x) \geq \frac12 x\ln x$ for $x \geq x_0$ with $x_0$ sufficiently large, and
that we can choose $\alpha, \beta$ such that $\beta/\alpha \geq x_0$.
Under the assumption {\bf(a-2)} we see that 

\begin{equation}\label{eq:a2}
\begin{array}{rcl} 
\Pee( \text{\bf(a)} )
& = & \Pee( \Bi( \delta k^2, \alpha^j ) \geq 1 ) \\
& \leq & \delta k^2 \alpha^j \\
& = & (\delta k^2 \beta^j ) \cdot (\alpha/\beta)^j \\
& \leq & (\alpha/\beta)^j,
\end{array} 
\end{equation}

\noindent 
using assumption {\bf(a-2)} for the third line.

Combining the above bounds we see that, provided $\alpha, \beta, \gamma$ were chosen appropriately, we have

\[ \Pee( E_{i,j} ) \leq 2 e^2 (\beta/\gamma^2)^j + (\alpha/\beta)^{j/2} < 10^{-j}, \]

\noindent
the last inequality again holding if we chose---as we may assume we have---$\alpha, \beta, \gamma$
appropriately.

We now consider the situation where $k \geq \sqrt{n}$.
Note that in this case, having chosen the constants appropriately, 

\[ k\gamma^j \geq \sqrt{n\gamma^{2I}} 
\geq \sqrt{n \gamma^{4\eps\log(n)}} = \Omega( n^{1/4} ).\]

Thus, the second line of~\eqref{eq:b} gives in fact that
\[ 
\Pee( \text{{\bf(b)}} ) = \exp\left[ - \Omega\left( n^{1/4} \right) \right]  
\]
Similarly, having chosen the constants correctly, we have 

\[ k^2 \beta^j > n \beta^I \geq n^{1 +\eps\log(\beta)} = \Omega( \sqrt{n} ). \]

\noindent
Thus, the case {\bf(a-2)} does not apply, and the third line of~\eqref{eq:a1} gives that 

\[ \Pee( \text{{\bf(a)}} ) = \exp[ - \Omega(\sqrt{n} ) ]. \]

\noindent
This shows that indeed $\Pee(E_{i,j}) = \exp[ - \Omega(n^{1/4}) ]$ in the case when 
$k \geq \sqrt{n}$.
\end{proof}

This last lemma is all we need to prove part~\ref{itm:ngeqk} of Theorem~\ref{thm:main}, for the case when $k \geq \sqrt{n}$.

\vspace{12pt}

\begin{proofof}{part~\ref{itm:ngeqk} of Theorem~\ref{thm:main}, assuming $k \geq \sqrt{n}$}
By the previous lemma, the probabilty that during any round of the process some group obtains a 
 non-increasing sequence of length $k \gamma^{j}$ is at most

\[ 
\Pee(\text{desired sequence is obtained} ) 
\leq 
\sum_{i+j \leq I} \Pee( E_{i,j} ) 
= 
O\left( (\log n)^2 \right) \cdot \exp[-\Omega(n^{1/4})] 
= o(1). 
\]

\noindent
This means that, with probability $1-o(1)$, there will be a total of $I = \Theta\left( \log(n/k) \right)$ rounds and in each 
round we spend the full amount of time $\delta k^2$. 
So indeed, with probability $1-o(1)$, the process does not complete before $\delta k^2 I = \Omega( k^2 \log(n/k) )$.
\end{proofof}

\vspace{12pt}

In the rest of the section we will thus be assuming that $k \leq \sqrt{n}$. Note that this in particular
implies that $I = \Theta( \log(n/k) ) \to \infty$.
We will let $J_i$ be the number of groups that get filled in in round $i$.

% \begin{lemma}
% There are universal constants $C < \infty$ and $c > 0$ such that, provided $\alpha, \beta, \gamma, \delta, \eps$ are chosen appropriately and 
% $m \gamma^i \leq n \alpha^{i}$,  we have $\Ee J_i^2 \leq C$ and 
% $\Pee( J_i = 1 ) \geq c$. 
% \end{lemma}
% 
% \begin{proof}
% We note that by Lemma~\ref{lem:10}, for $j\geq 2$, we have 
% 
% \[ \Pee( J_i \geq j ) \leq \sum_{k\geq j} 10^{-k} = \frac{10^{1-j}}{9}. \]
% 
% \noindent 
% (The case $j=1$ is different, because $J_i=1$ also holds if all the events $E_{i,j}$ fail.)
% % so that 
% % 
% % \[ \Ee J_i = \sum_{j\geq 1} \Pee( J_i \geq j ) \leq 1 + \sum_{j\geq 2} \frac{10^{1-j}}{9} = \frac{82}{81}. \]
% % 
% % \noindent 
% % Similarly we have 
% This gives that 
% 
% \[ \Ee J_i^2 = \sum_{j\geq 1} \Pee( J_i^2 \geq \sqrt{j} )
% \leq 1 + \sum_{j\geq 2} \frac{10^{1-\sqrt{j}}}{9} =: C < \infty. \]
% 
% \noindent
% Similarly we find: 
% 
% \[ \begin{array}{rcl} 
% \Pee( J_i = 1 ) 
% & = & 1 - \sum_{j=2}^\infty \Pee( J_i=j )  \\
% & \geq & 1 - \sum_{j=2}^\infty \sum_{k=j}^\infty \Pee( E_{i,j} ) \\
% & \geq & 1 - \sum_{j=2}^\infty \sum_{k=j} 10^{-k} \\
% & = & 1 - \frac{1}{81}.
% \end{array} \]
% 
% \end{proof}

\vspace{12pt}

\begin{proofof}{part~\ref{itm:ngeqk} of Theorem~\ref{thm:main}, assuming that $k \leq \sqrt{n}$}
%Note that $\Ee J_i \leq \sqrt{C}$ with $C$ as provided by Lemma~\ref{lem:EJ2}.
%Let us set $I' := I / 1000 C$.
Appealing to Lemma~\ref{lem:10}, we see that sequence $J_1, J_2, \dots$ is stochastically dominated by an 
i.i.d.~sequence $J_1', J_2', \dots$ where $\Pee( J_1' = j ) = 10^{-j}$ for $j\geq 2$
and $\Pee( J_1' = 1 ) = \frac{89}{90}$.
It is easily checked that $\Ee (J_1')^2 < \infty$.
Let us set $I' := \eps I$ for some suitable chosen $0 < \eps < 1$.
By the law of large number (which holds for i.i.d.~sequences with a finite second moment) we have that 

\[ \Pee( J_1'+\dots+J_{I'}'> 2 I' \Ee (J_1') ) = o(1). \]

\noindent
Thus, choosing $\eps$ appropriately (namely, so that $2\eps \Ee (J_1')< 1$), 
we have 

\[ \Pee( J_1 + \dots + J_{I'} > I ) 
 \leq \Pee( J_1'+\dots+J_{I'}' > I ) = o(1).
\]

\noindent
In other words, with probability $1-o(1)$, there are at least $I'$ rounds.

Now let $Y_i$ be a $\{0,1\}$-valued variable that equals one if in round $i$ no group obtained a non-increasing sequence of length $k \gamma^{j}$.
Observe that 

\[ \Pee( Y_i = 1 ) \geq 1 - \sum_{j=1}^{I-i} \Pee( E_{i,j} ) 
 \geq 1 - \sum_{j=1}^\infty 10^{-j}
 = \frac{8}{9}.
\]

\noindent
Observe that $Y_1, Y_2, \dots$ stochastically dominates an i.i.d~sequence $Y_1', Y_2', \dots$ 
of Bernoulli($\frac89$) random variables.
Again by the law of large numbers, we have that 

\[ \Pee( Y_1 +\dots+ Y_{I'} > \frac49 I' ) \geq 
 \Pee( Y_1'+\dots+Y_{I'}' > \frac49 I' ) = 1 - o(1).
\]

\noindent
So this shows that with probability $1-o(1)$ there were at least $\frac49 I' = \Omega( log(n/k) )$ rounds in which no group obtains a non-increasing sequence of length $k \gamma^{j}$, which means that the round lasted the full amount of $\delta k^2$ time.
Hence, with probability $1-o(1)$, at time $\delta k^2 I' = \Omega( k^2 \log( n/k ) )$ the process has not yet completed.
\end{proofof}

%\appendix

\subsection{A version of the Chernoff bound for sums of independent binomials and sums of independent geometrics}

We first briefly recall the version of the Chernoff bound for binomials and Poisson-distributed random variables that we use above.
A proof can be found, for instance, in~\cite{penroseboek}.

\begin{lemma}[Chernoff bound]\label{lem:chernoffvanilla}
Let $X$ be %either 
a binomial %or a Poisson 
random variable.
\begin{enumerate}
 \item For all $t \geq \Ee (X)$ we have $\Pee( X \geq t ) \leq e^{- \Ee (X) \cdot H( t/\Ee (x) ) }$, 
 \item For all $t \leq \Ee (X)$ we have $\Pee( X \leq t ) \leq e^{- \Ee (X) \cdot H( t/\Ee (x) ) }$,
\end{enumerate}
\noindent
where $H(x) := x\ln x - x + 1$.
\end{lemma}

A variation of the proof of this last lemma gives:

\begin{lemma}\label{lem:chernoffsumbin}
Suppose that $X = X_1+\dots+X_k$ is a sum of independent binomials.
Then 

\[ \Pee( X \geq t ) \leq e^{-\Ee (X) \cdot H( t / \Ee (X) )}, \]

\noindent
for all $t \geq \Ee (X)$, where $H(x) := x\ln x - x + 1$.
\end{lemma}

\begin{proofsketch}
\noindent
Note that $\Ee (z^X) = \Pi_{i=1}^k \Ee (z^{X_i})$ by independence, and that 

\[ \Ee (z^X_i) = \left(1 + p_i(z-1) \right)^{n_i} \leq e^{(z-1)n_ip_i} = e^{(z-1)\Ee (X_i)}, \]

\noindent
so we find that $\Ee (z^X) \leq e^{(z-1)\Ee (X)}$.
We can thus proceed as in the proof of Lemma 1.2 in Penrose's book~\cite{penroseboek}.
%\hfill $\blacksquare$
\end{proofsketch}

\begin{lemma}\label{lem:chernoffsumgeo}
 Suppose that $X = X_1 + \dots + X_k$ is a sum of i.i.d.~geometrically distributed random variables.
 \begin{enumerate}
  \item For all $t \geq \Ee (X)$ we have $\Pee( X \geq t ) \leq e^{ - (kt/\Ee (X)) \cdot H( \Ee (X) / t ) }$,
  \item For all $t \leq \Ee (X)$ we have $\Pee( X  \leq t ) \leq e^{ - (kt/\Ee (X)) \cdot H( \Ee (X) / t ) }$.
 \end{enumerate}
\end{lemma}

\begin{proof}
We first consider $\Pee( X \geq t )$.
Let $Y \isd \Bi(t,p)$, where $p$ is of course the common success probability 
of the $X_i$-s. 
We observe that $\Pee( X \geq t ) = \Pee( Y \leq k )$ (the probability one has to make $t$ or more tries to get $k$ successes equals the 
probability that in a sequence of $t$ coin flips $k$ or fewer successes occur).
Now we apply the plain vanilla Chernoff bound to obtain:

\[ \Pee( X \geq t ) \leq \exp[ - \Ee (Y) \cdot H( k/\Ee (Y) ) ]  = \exp[ - tp \cdot H( k/tp ) ]. \]

\noindent
Observing that $\Ee (X) = k \Ee (X_1) = k/p$, we see that $tp = k t / \Ee (X)$ and $k/tp = \Ee (X) / t$.

The case of $\Pee( X \leq t )$ is proved completely analogously.
\end{proof}

\section{Delay-based activation orders}
\label{sec:delayed_analysis}

In order to be able to include delay-based protocols in our study, we abstract away from the low-level implementation details, and view them as existing protocols (local decision rules) enhanced with delay functions.

All protocols described in Section~\ref{sec:delay_methods}  include local decision rules that state what to do if a node receives a message already in its queue.
One option, is to always cancel the transmission of a message received twice. This is what is done, for instance, in BLR. Note that, disregarding the delay, this would be equivalent to the M heuristic with $M=2$.
Another option considered is to cancel only if the sender of the duplicate message is closer to the destination than the current node. Both GeRaF and Greedy routing apply rules of this type. Again, disregarding the delay, this would be similar to the CD heuristic.

In heuristics where the lower bound and upper bounds on \recmess are tight regardless of the activation order, such as simple flooding, the M and T heuristics, the choice of a  delay function will not influence \recmess by more than a factor two.

In heuristics where these bounds  are not right, such as CD and CD-P, the delay function can make a significant difference.
For instance, consider a delay-based protocol where the nodes delete messages from their queues only when they receive a duplicate from a node that is closer to the destination, combined with  delay function  $\text{delay}_2  = \text{MD} \frac{p}{r} $ (recall that MD is a constant representing the maximum delay).
For simplicity of analysis, we assume without loss of generality $\text{MD}=r$, which implies that the delay increases by exactly one time step per node.
Then, the bounds for \recmess in the unbounded scenario become:

\[
\recpac=
\begin{cases}
2^k & \text{ if } k<\log n\,,\\
n+n(k-\log n) & \text{ if } k\geq\log n\,.
\end{cases}
\]

%\begin{proof}
%\rodrigo{proof?}
%\end{proof}

%\rodrigo{What about the bounded case?}

Analyzing all combinations of local decisions rules and delay functions would be interesting, although it falls outside the scope of this work.

\section{Discussion}
\label{sec:conclusions}
% Concluding remarks

Beaconless geocast protocols are used in practice in 2D scenarios. They differ from the 1D ones in a few but important characteristics, most notably that obstacles (like buildings, which cannot be traversed by the transmission signal) need to be surrounded, and that local optimization strategies fail to guarantee delivery. Therefore, combinations of different strategies need to be used in order to achieve delivery guarantees and, at the same time, keep the network load within reasonable bounds. The network load analysis in this cases is difficult, and almost only experimental results exist. 
This motivated studying the 1D case.
We have shown that the rigorous analysis of geocast protocols even in simple 1D scenarios can be interesting---and challenging. 
Indeed, all protocols give rise to different load bounds, with CD and CD-P being particularly subtle to analyze due to the fact that an involved probabilistic analysis is required.

Our theoretical analysis confirms behaviors that had been observed before only through simulations. 
The different expression for \recpac, summarized in Tables~\ref{tab:results_worst_case} and~\ref{tab:results_probabilistic}, make evident the differences between flooding and the restricted flooding (M- and T-heuristics) and delay-based protocols, as well as the advantage of CD over them for small values of $k$.
Moreover, our analysis gives theoretical support to the performance claims of CD-P: Hall~\cite{h-igmanet-11} showed through simulations that CD-P scales better and improves its efficiency over CD considerably.
Our results indeed corroborate this. In the unbounded reach scenario, CD has serious scalability problems when the number of messages ($k$) approaches the number of nodes ($n$), making it behave like flooding, while CD-P's dependence on $k$ remains always linear, and, more importantly, only logarithmic on $n$. 
For the bounded reach scenario, CD again shows an overhead in $k$ (this time of only $O(\sqrt {k})$), while CD-P achieves the optimal asymptotic performance of $\Theta(k)$.

The results in this paper are a first step towards analyzing geocast protocols in generic geometric settings.
This paper has focused on two relatively simple scenarios.
Between these basic scenarios and the final intricacies of real-world situations, several abstractions of the geocast problem of varying complexity can be imagined, which
are definitely worth studying in the future (see Figure~\ref {fig:scenarios2}).

\begin{figure}[t]
  \centering
    \includegraphics[width=0.825\textwidth]{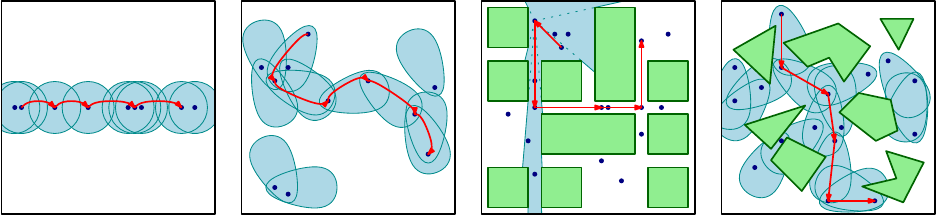}
  \caption{Scenarios of increasing complexity. (i)~1D scenario. (ii)~2D scenario; non-uniform ranges. (iii)~Regular obstacles; visibility. (iv)~Irregular obstacles; combining visibility and non-uniform ranges.}
\label{fig:scenarios2}
\end{figure}

%\section* {Acknowledgements}

%This work was partially supported by the Netherlands Organisation for Scientific Research (NWO) under project no. 614.001.504.

\bibliographystyle{abbrv}
\bibliography{castnotes}
\end {document}